\documentclass[12pt,onecolumn,draftclsnofoot]{IEEEtran}
\IEEEoverridecommandlockouts

\usepackage{times}
\usepackage[final]{graphicx}
\usepackage{amsmath,amsfonts}
\usepackage{amsthm}
\usepackage{amssymb,amsbsy}

\usepackage{cite}
\usepackage{multirow}

\newcommand{\ignore}[1]{}

\newcommand{\hsp}{\hspace{0.1in} }
\newcommand{\hspp}{\hspace{0.05in} }
\newcommand{\hsppp}{\hspace{0.02in} }

\newtheorem{prop}{Proposition}

\newsavebox{\savepar}

\begin{document}
\title{Directional Beamforming for Millimeter-Wave MIMO Systems}
\author{\large Vasanthan Raghavan, Sundar Subramanian, Juergen Cezanne, and Ashwin Sampath
\thanks{The authors are with Qualcomm Flarion Technologies, Inc., Bridgewater,
NJ 08807, USA.
E-mail: \{vraghava, sundars, jcezanne, asampath@qti.qualcomm.com\}.}
}

\maketitle
\vspace{-10mm}

\begin{abstract}
\noindent
The focus of this paper is on beamforming in a millimeter-wave (mmW)
multi-input multi-output (MIMO) setup that has gained increasing traction
in meeting the high data-rate requirements of next-generation wireless
systems. For a given MIMO channel matrix, the optimality of beamforming
with the dominant right-singular vector (RSV) at the transmit end and with
the matched filter to the RSV at the receive end has been well-understood.
When the channel matrix can be accurately captured by a physical (geometric)
scattering model across multiple clusters/paths as is the case in mmW MIMO
systems, we provide a physical interpretation for this optimal structure: beam
steering across the different paths with appropriate power allocation and phase
compensation. While such an explicit physical interpretation has not been
provided hitherto, practical implementation of such a structure in a mmW
system is fraught with considerable difficulties (complexity as well as cost)
as it requires the use of per-antenna gain and phase control. This paper
characterizes the loss in received ${\sf SNR}$ with an alternate low-complexity
beamforming solution that needs only per-antenna phase control and corresponds
to steering the beam to the dominant path at the transmit and receive ends. While
the loss in received ${\sf SNR}$ can be arbitrarily large (theoretically), this
loss is minimal in a large fraction of the channel realizations reinforcing the
utility of directional beamforming as a good candidate solution for mmW MIMO systems.
\end{abstract}

\section{Introduction}
The ubiquitous nature of communications made possible by the smart-phone and
social media revolutions has meant that the data-rate requirements will continue
to grow at an exponential rate. On the other hand, even under the most optimistic
assumptions, system resources can continue to scale at best at a linear rate leading
to enormous mismatches between supply and demand. Given this backdrop, many candidate
solutions have been proposed~\cite{rusek,qualcomm,boccardi} to mesh into the
patchwork that addresses the $1000$-${\sf X}$ data challenge~\cite{qualcomm1} --- an
intermediate stepping stone towards bridging this burgeoning gap.

One such solution that has gained increasing traction over the last few years is
communications over the millimeter-wave (mmW) regime~\cite{khan,rappaport,rangan,roh}
where the carrier frequency is in the $30$ to $300$ GHz range. Spectrum crunch, which
is the major bottleneck at lower/cellular carrier frequencies, is less problematic at higher
carrier frequencies due to the availability of large (either unlicensed or lightly
licensed) bandwidths. However, the high frequency-dependent propagation and shadowing
losses (that can offset the link margin substantially) complicate the exploitation of
these large bandwidths. It is visualized that these losses can be mitigated by
limiting coverage to small areas and leveraging the small wavelengths that allows the
deployment of a large number of antennas in a fixed array aperture.

Despite the possibility of multi-input multi-output (MIMO) communications, mmW
signaling differs significantly from traditional MIMO architectures at cellular
frequencies. The most optimistic antenna configurations\footnote{In a downlink setting,
the first dimension corresponds to the number of antennas at the user equipment end and
the second at the base-station end.} at cellular frequencies
are on the order of $4 \times 8$ with a precoder rank (number of layers) of $1$ to
$4$; see, e.g.,~\cite{lim}. Higher rank signaling requires multiple radio-frequency (RF)
chains\footnote{An RF chain includes (but is not limited to) analog-to-digital
and digital-to-analog converters, power and low-noise amplifiers, mixers, etc.}
which are easier to realize at lower frequencies than at the mmW regime. Thus,
there has been a growing interest in understanding the capabilities of
low-complexity approaches such as beamforming (that require only a single RF chain)
in mmW systems~\cite{venkateswaran,torkildson,brady,oelayach,hur,alkhateeb}.

On the other hand, smaller form factors at mmW frequencies ensure\footnote{For
example, a $64$ element uniform linear array (ULA) at $30$ GHz requires an
aperture of $\sim 1$ foot at the critical $\lambda/2$ spacing --- a constraint
that can be realized at the base-station end.} that configurations such as $4
\times 64$ are realistic. Such high antenna dimensionalities as well as the
considerably large bandwidths at mmW frequencies result in a higher resolvability
of the multipath and thus, the MIMO channel is naturally sparser in the mmW regime
than at cellular frequencies~\cite{vasanth_jstsp,vasanth_it2,raghavan_ett}. In
particular, the highly directional
nature of the channel ensures the relevance of physically-motivated beam steering
at either end, which is difficult (if not impossible) at cellular frequencies. While
this physical connection has been implicitly and intuitively understood, an explicit
characterization of this connection has remained absent so far.

We start with such an explicit physical interpretation in this work by showing that
the optimal beamformer structure corresponds to beam steering across the different
paths that capture the MIMO channel with appropriate power allocation and phase
compensation. We also illustrate the structure of this power allocation and phase
compensation in many interesting special cases. Despite using only a single RF chain,
the optimal beamformer requires per-antenna phase and gain control (in general),
which could render this scheme disadvantageous from a cost perspective. Thus, we
study the loss in received ${\sf SNR}$ with a simpler scheme that requires only
phase control and steers beams to the dominant path at either end. Our study shows
that this simpler scheme suffers only a minimal loss relative to the optimal
beamforming scheme in a large fraction of the channel realizations, thus making it
attractive from a practical standpoint.

\noindent {\bf \em \underline{Notations:}} Lower- (${\bf x}$) and upper-case block
(${\bf X}$) letters denote vectors and matrices with ${\bf x}(i)$ and ${\bf X}(i,j)$
denoting the $i$-t
h and $(i,j)$-th entries of ${\bf x}$ and ${\bf X}$, respectively.
$\| {\bf x} \|_2$ 
denotes the $2$-norm 
of a vector ${\bf x}$ (that is, $\| {\bf x} \|_2 = \left(
\sum_i |{\bf x}(i)|^2 \right)^{1/2}$), 
whereas ${\bf x}^H$ and
${\bf x}^T$ denote the complex conjugate Hermitian and regular transposition
operations of ${\bf x}$, respectively. We use ${\mathbb{Z}}$, ${\mathbb{R}}$,
${\mathbb{R}}^+$ and ${\mathbb{C}}$ to denote the field of integers, real numbers,
positive reals and complex numbers, respectively.

\section{System Setup}
Let ${\bf H}$ denote the $N_r \times N_t$ channel matrix with $N_r$ receive and $N_t$ transmit
antennas. We assume an extended Saleh-Valenzuela geometric model~\cite{saleh}
for the channel where ${\bf H}$ is determined by scattering over $L$ clusters\footnote{Each
cluster is assumed to have one dominant path and diffuse scattering over a cluster with
multiple sub-paths is not captured here.} and denoted as follows:
\begin{eqnarray}
{\bf H} = \sqrt{ \frac{ N_r N_t}{L} } \cdot
\sum_{\ell = 1}^L \alpha_{\ell} \cdot {\bf u}_{\ell} \hsppp {\bf v}_{\ell}^H
\label{eq_chan_model}
\end{eqnarray}
where $\alpha_{\ell} \sim {\cal CN}(0,1)$ denotes the complex gain, ${\bf u}_{\ell}$ denotes
the $N_r \times 1$ receive array steering vector, and ${\bf v}_{\ell}$ denotes the
$N_t \times 1$ transmit array steering vector, all corresponding to the $\ell$-th
path. 
With this assumption, the normalization constant
$\sqrt{ \frac{ N_r N_t}{L} }$ in ${\bf H}$ ensures that the standard channel power
normalization in MIMO system studies holds.
As a typical example of the case where a uniform linear array (ULA) of antennas are
deployed at both ends of the link (and without loss of generality pointing along the
X axis), the array steering vectors ${\bf u}_{\ell}$ and ${\bf v}_{\ell}$ corresponding
to angle of arrival (AoA) $\phi_{ {\sf R}, \ell}$ and angle of departure (AoD)
$\phi_{ {\sf T}, \ell}$ in the azimuth (assuming an elevation angle
$\theta_{ {\sf R}, \ell} = \theta_{ {\sf T}, \ell} = 90^{o}$) are given as
\begin{eqnarray}
{\bf u}_{\ell} & = & \frac{1}{\sqrt{N_r}} \cdot
\left[1 \hspp e^{j k d_{\sf R} \cos(\phi_{ {\sf R}, \ell})}
\hspp e^{j 2 k d_{\sf R} \cos(\phi_{ {\sf R}, \ell})}  \hspp \cdots \hspp
e^{j (N_r - 1) k d_{\sf R} \cos(\phi_{ {\sf R}, \ell})}
\right]^T 
\label{eq_ul}
\\
{\bf v}_{\ell} & = & \frac{1}{\sqrt{N_t}} \cdot \left[1 \hspp e^{j k d_{\sf T} \cos(\phi_{ {\sf T}, \ell})}
\hspp e^{j 2 k d_{\sf T} \cos(\phi_{ {\sf T}, \ell})}  \hspp \cdots \hspp
e^{j (N_t - 1) k_{\sf T} d \cos(\phi_{ {\sf T}, \ell})}
\right]^T
\label{eq_vl}
\end{eqnarray}
where $k = \frac{ 2 \pi}{\lambda}$ is the wave number with $\lambda$ the wavelength of
propagation, and $d_{\sf R}$ and $d_{\sf T}$ are the inter-antenna element spacing at the
receiver and transmitter sides, respectively. To simplify the notations and to
capture the {\em constant phase offset} (CPO)-nature of the array-steering vectors and
the correspondence with their respective physical angles, we will
henceforth\footnote{Similar notation will also be followed for other vectors with a
constant phase offset across the array.} denote ${\bf u}_{\ell}$ and
${\bf v}_{\ell}$ in~(\ref{eq_ul})-(\ref{eq_vl}) as ${\sf CPO} (\phi_{ {\sf R}, \ell})$
and ${\sf CPO}(\phi_{ {\sf T}, \ell})$, respectively. With the typical $d_{\sf R} =
d_{\sf T} = \frac{\lambda}{2}$ spacing, we have $k d_{\sf R} = k d_{\sf T} = \pi$. In
the general case where the paths depart at an AoD pair of $( \theta_{ {\sf T}, \ell},
\hsppp \phi_{ {\sf T}, \ell})$ and arrive at an AoA pair of $( \theta_{ {\sf R}, \ell},
\hsppp \phi_{ {\sf R}, \ell})$ in the elevation and azimuth, respectively, the
$\cos(\phi_{ {\sf R}, \ell})$ and $\cos(\phi_{ {\sf T}, \ell})$ terms
in~(\ref{eq_ul})-(\ref{eq_vl}) are replaced with $\sin(\theta_{ {\sf R}, \ell})
\cos(\phi_{ {\sf R}, \ell})$ and $\sin(\theta_{ {\sf T}, \ell}) \cos(\phi_{ {\sf T},
\ell})$, respectively. Similar expressions for ${\bf u}_{\ell}$ and ${\bf v}_{\ell}$
can be written if the array is placed on the Y or Z axes or with a planar
array; see~\cite{balanis,oelayach}, for example.

We are interested in beamforming (rank-$1$ signaling) over ${\bf H}$ with the unit-norm
$N_t \times 1$ beamforming vector ${\bf f}$. The system model in this setting is given as
\begin{eqnarray}
{\bf y} = \sqrt{\rho_{\sf prebf}} \cdot {\bf H} {\bf f} s + {\bf n}
\label{eq_sys_model}
\end{eqnarray}
where $\rho_{\sf prebf}$ is the pre-beamforming ${\sf SNR}$, $s$ is the symbol chosen
from an appropriate constellation for signaling, and ${\bf n}$ is the $N_r \times 1$
proper complex white Gaussian noise vector (that is, ${\bf n} \sim {\cal CN}({\bf 0}, \hsppp
{\bf I})$) added at the receiver. The symbol $s$ is decoded by beamforming at the
receiver along the unit-norm $N_r \times 1$ vector ${\bf g}$ to obtain
\begin{eqnarray}
\widehat{s} = {\bf g}^H {\bf y} = \sqrt{\rho_{\sf prebf}} \cdot {\bf g}^H
{\bf H} {\bf f} s + {\bf g}^H {\bf n}.
\label{eq_decoded_symbol}
\end{eqnarray}

Let ${\cal F}_2$ denote the class of energy-constrained beamforming vectors.
That is, ${\cal F}_2 = \{ {\bf f} \hsppp : \hsppp \| {\bf f} \|_2 \leq 1 \}$.
Under perfect channel state information (CSI) (that is, ${\bf H} = {\sf H}$)
at both the transmitter and the receiver, optimal beamforming vectors
${\bf f}_{\sf opt}$ and ${\bf g}_{\sf opt}$ are to be designed
from ${\cal F}_2$ to maximize 
the received ${\sf SNR}$~\cite{tky_lo}, defined as,
\begin{eqnarray}
{\sf SNR}_{\sf rx} \triangleq \rho_{\sf prebf} \cdot
\frac{ | {\bf g}^H \hsppp {\sf H} \hsppp {\bf f}|^2 \cdot {\bf E} [|s|^2] }
{ {\bf E} \left[ |{\bf g}^H i{\bf n}|^2 \right] } =
\rho_{\sf prebf} \cdot \frac{ | {\bf g}^H \hsppp {\sf H} \hsppp {\bf f}|^2 }
{ {\bf g}^H {\bf g} }.
\nonumber
\end{eqnarray}
Clearly, the above quantity is maximized with $\| {\bf f}_{\sf opt} \|_2 = 1$,
otherwise energy is unused in beamforming. Further, a simple application of
Cauchy-Schwarz inequality shows that ${\bf g}_{\sf opt}$ is a matched filter
combiner at the receiver with $\| {\bf g}_{\sf opt} \|_2 = 1$ resulting in
${\sf SNR}_{\sf rx} = \rho_{\sf prebf} \cdot {\bf f}^H \hsppp {\sf H}^H
\hsppp {\sf H} \hsppp {\bf f}$. 
We thus have
\begin{eqnarray}
{\bf f}_{\sf opt} = {\sf v}_1({\sf H}^H {\sf H}), \hsp
{\bf g}_{\sf opt} = \frac{ {\sf H} \hsppp  {\sf v}_1({\sf H}^H {\sf H}) }
{ \|  {\sf H} \hsppp {\sf v}_1({\sf H}^H {\sf H}) \|_2 } ,
\label{eq_rsv_mf}
\end{eqnarray}
where ${\sf v}_1({\sf H}^H {\sf H})$ denotes a dominant unit-norm right singular
vector of ${\sf H}$. Here, the singular value decomposition of ${\sf H}$ is given as
${\sf H} = {\sf U} \hsppp {\sf \Lambda} \hsppp {\sf V}^H$
with ${\sf U}$ and ${\sf V}$ being $N_r \times N_r$ and $N_t \times N_t$ unitary
matrices of left and right singular vectors, respectively, and arranged so that the
corresponding leading diagonal entries of the $N_r \times N_t$ singular value matrix
${\sf \Lambda}$ are in non-increasing order.

\section{Explicit connection between ${\bf f}_{\sf opt}$, ${\bf g}_{\sf opt}$
and physical directions}
\label{app_B}
A typical {\em sparse} mmW channel can be assumed to consist of a small number of dominant
clusters (say, $L = 2$ or
$3$)~\cite{rangan,vasanth_jstsp,vasanth_it2,raghavan_ett,samimi_conf}. For example, a
dominant line-of-sight (LOS)
path with strong reflectors in the form of a few glass windows of buildings in the vicinity
of the transmitter or the receiver could capture an urban mmW setup. In the context of such
a sparse mmW channel ${\sf H}$, the intuitive meaning of ${\bf f}_{\sf opt}$ is to
``coherently combine'' (by appropriate phase compensation) the energy across the multiple
paths so as to maximize the energy delivered to the receiver. The precise connection between
the physical directions $\{ \phi_{{\sf R}, \ell}, \hsppp \phi_{ {\sf T}, \ell} \}$ in the
ULA channel model and ${\bf f}_{\sf opt}$ in~(\ref{eq_rsv_mf}) is established 
next. Towards this goal, a preliminary result is established first.

\begin{prop}
\label{prop_evectors_HhermH}
With ${\bf H} = {\sf H}$ and the channel model in~(\ref{eq_chan_model}), all the
eigenvectors of ${\sf H}^H {\sf H}$ can be represented as linear combinations
of ${\bf v}_1, \cdots , {\bf v}_L$.
\end{prop}
\begin{proof}
See Appendix~\ref{app_prop_evectors_HhermH}.
\end{proof}
It is important to note that while the right singular vectors of ${\sf H}$ (also, the
eigenvectoirs of ${\sf H}^H {\sf H}$) are orthonormal by construction,
${\bf v}_1, \cdots, {\bf v}_L$ need not be orthonormal. With this background,
Prop.~\ref{prop_evectors_HhermH} provides a non-unitary basis for the eigen-space of
${\sf H}^H {\sf H}$ when $L \leq N_t$. 
As another ramification of this fact, in the case where $L > N_t$, it is still true that
the set $\{ {\bf v}_1, \cdots , {\bf v}_L \}$ spans the eigen-space of ${\sf H}^H {\sf H}$,
however this set is no longer a basis. These facts along with the fact that
${\bf f}_{\sf opt}$ is a dominant eigenvector of ${\sf H}^H {\sf H}$ also implies the
following:
\begin{eqnarray}
\frac{ {\sf SNR}_{\sf rx} } {\rho_{\sf prebf} } =
\max \limits_{ {\bf f} \hsppp : \hsppp \|{\bf f}\|_2 = 1 }
\frac{ {\bf f}^H \hsppp {\sf H}^H \hsppp {\sf H} \hsppp {\bf f} }
{ {\bf f}^H \hsppp {\bf f} } =
\max \limits_{ {\bf f} \hsppp : \hsppp {\bf f} \hsppp \in \hsppp {\cal G}_0 }
\frac{ {\bf f}^H \hsppp {\sf H}^H \hsppp {\sf H} \hsppp {\bf f} }
{ {\bf f}^H \hsppp {\bf f} }
\label{eq_eqviv_optimization}
\end{eqnarray}
where ${\cal G}_0 \triangleq \left\{ {\bf f} \hsppp : \hsppp {\bf f} = \sum_{i = 1}^L
e^{j \theta_i} \hsppp \beta_i \hsppp {\bf v}_i \right\}$ with $\beta_i \in
{\mathbb{R}}^+,  \hsppp \theta_i \in [0, \hsppp 2 \pi), \hsppp i = 1, \cdots, L$.
Without loss in generality, we can set $\beta_L = \sqrt{ 1 - \sum_{j=1}^{L-1}
\beta_j^2 }$ and $\theta_1 = 0$ in the definition of ${\cal G}_0$ to reduce the
optimization in~(\ref{eq_eqviv_optimization}) to a $2(L-1)$-dimensional optimization
over ${\cal G}$, defined as,
\begin{eqnarray}
{\cal G} \triangleq
\left\{ {\bf f} \hsppp : \hsppp {\bf f} = \beta_1 \hsppp {\bf v}_1
+ \sum_{i = 2}^{L-1} e^{j \theta_i} \hsppp \beta_i \hsppp {\bf v}_i
+ e^{j \theta_L } \hsppp \sqrt{1 - \sum_{j=1}^{L-1} \beta_j^2 } \hsppp {\bf v}_L
\right\}. \nonumber
\end{eqnarray}
In other words, the optimization over the space of ${\cal G}$
should result in the dominant right singular vector of ${\sf H}^H {\sf H}$. Note
that the constraint set in the optimization over ${\cal G}$ is the outer product
of a $(L-1)$-dimensional real sphere where $\sum_{i=1}^{L-1} \beta_i^2 \leq 1$
with a cuboid $\prod_{i = 2}^L \theta_i, \hsppp 0 \leq \theta_i < 2 \pi$.

We now consider the special case where $L = 2$ and perform this optimization
and thus provide a physical interpretation of 
${\bf f}_{\sf opt}$. For this, note that in the $L = 2$ case, 
${\sf H}^H {\sf H}$ simplifies to
\begin{eqnarray}
\frac{L}{N_t N_r} \cdot {\sf H}^H {\sf H} =
|\alpha_1|^2 \cdot {\bf v}_1 {\bf v}_1^H +
|\alpha_2|^2 \cdot {\bf v}_2 {\bf v}_2^H +
\alpha_1^{\star} \alpha_2 \cdot ( {\bf u}_1^H {\bf u}_2) \cdot {\bf v}_1 {\bf v}_2^H
+ \alpha_2^{\star} \alpha_1 \cdot ( {\bf u}_2^H {\bf u}_1) \cdot {\bf v}_2 {\bf v}_1^H.
\nonumber
\end{eqnarray}
With $\beta_1 = \beta$ and $\theta_2 = \theta$ in the archetypical ${\bf f}$ from ${\cal G}$,
the norm of ${\bf f}$ is given as
\begin{eqnarray}
{\bf f}^H {\bf f}
= 1 + 2 \beta \sqrt{1 - \beta^2}  \cdot \ | {\bf v}_1^{H} {\bf v}_2 |
\cdot \cos( \phi) 
\label{den_rx_snr}
\end{eqnarray}
where $\phi \triangleq \theta + \angle{{\bf v}_1^{H} {\bf v}_2}$. Further, a
tedious but straightforward calculation shows that
\begin{eqnarray}
\begin{split}
& {\hspace{0.1in}} \frac{L}{N_t N_r} \cdot
{\bf f}^{H} {\sf H}^{H} {\sf H} {\bf f}
\nonumber \\
& {\hspace{0.15in}} =
\beta^2 |\alpha_1|^2 + (1 - \beta^2) |\alpha_2|^2 +
\left( \beta^2 |\alpha_2|^2 + (1 - \beta^2) |\alpha_1|^2 \right)
| {\bf v}_1^{H} {\bf v}_2|^2
\nonumber \\
& {\hspace{0.35in}}
+ 2 |\alpha_1| | \alpha_2 | 
\cdot | {\bf v}_1^{H} {\bf v}_2| \cdot | {\bf u}_1^{H} {\bf u}_2|
\cdot \cos \left( \nu
\right)
+ 2 \beta \sqrt{1 - \beta^2} \cdot \left( |\alpha_1|^2 + |\alpha_2|^2 \right)
\cdot | {\bf v}_1^{H} {\bf v}_2| \cdot \cos \left( \phi
\right) \nonumber \\
&  {\hspace{0.35in}}
+ 2 \beta \sqrt{1 - \beta^2} \cdot
 |\alpha_1 | | \alpha_2| \cdot | {\bf u}_1^{H} {\bf u}_2 |
\cdot \left[ | {\bf v}_1^{H} {\bf v}_2|^2 \cdot
\cos\left( \nu + \phi
\right) 
+ \cos \left( \nu - \phi
\right) \right]
\label{num_rx_snr}
\end{split}
\end{eqnarray}
where $\nu \triangleq \angle{ {\bf v}_1^{H} {\bf v}_2} - \angle{ {\bf u}_1^{H} {\bf u}_2}
+ \angle{ \alpha_1} - \angle{ \alpha_2 }$.
Observe that the phase term $\nu$ captures the {\em phase misalignment} between the
two paths since $| {\bf u}_i^H \hsppp {\sf H} \hsppp {\bf v}_j|$ is maximized for
all $\{ i,j\} \in 1,2$ when $\nu = 0$ (coherent phase alignment).

We now consider many special cases to study the performance of the beamforming scheme.
For this, we define the normalized received ${\sf SNR}$ (denoted
as $\widetilde{ {\sf SNR}}_{\sf rx}$):
\begin{eqnarray}
\widetilde{ {\sf SNR}}_{\sf rx} \triangleq
\frac{ {\sf SNR}_{\sf rx}} {N_t N_r \cdot \rho_{\sf prebf}} =
\frac{1}{N_t N_r} \cdot
\frac{ {\bf f}^{H} {\sf H}^{H} {\sf H} {\bf f} }{ {\bf f}^H {\bf f}}.
\nonumber
\end{eqnarray}

We start with a physical interpretation for the inner product between ${\bf u}_1$
and ${\bf u}_2$ (a similar interpretation holds for ${\bf v}_1^H {\bf v}_2$),
corresponding to CPO beams in two directions/paths. With the assumption
for ${\bf u}_{\ell}$ in~(\ref{eq_ul}), we have
\begin{eqnarray}
{\bf u}_1^H {\bf u}_2 = \exp\left( \frac{ j \pi(N_r - 1) \cdot \Delta \cos(\phi_{\sf R})} {2}
\right) \cdot \frac{ \sin \left( N_r \hsppp \pi \hsppp \Delta \cos(\phi_{\sf R}) /2 \right)}
{N_r \sin \left( \pi \hsppp \Delta \cos(\phi_{\sf R}) /2 \right)}
\label{eq_u1Hu2}
\end{eqnarray}
where $\Delta \cos(\phi_{\sf R}) \triangleq \cos(\phi_{ {\sf R},2}) -  \cos ( \phi_{ {\sf R},1})$.
Clearly, the maximum magnitude of 
${\bf u}_1^H {\bf u}_2$ is $1$ which is achieved when $\Delta \cos(\phi_{\sf R}) = 0$
(or when the two paths can be coherently combined
in the physical angle space). Further, a minimum magnitude of $0$ is achieved
in~(\ref{eq_u1Hu2}) when $\Delta \cos(\phi_{\sf R}) = \frac{2n}{N_r}, \hsppp n \in
{\mathbb{Z}} \backslash \{ 0 \}$. We denote this condition as {\em electrical orthogonality} between the
two paths, which is achievable with higher regularity in the physical angle space as $N_r$
increases.

\subsection{Beamforming along the dominant path}
We start with a scheme where the entire power is directed along {\em only} one
path (the dominant one): either ${\bf v}_1$ ($\beta = 1$) or ${\bf v}_2$ ($\beta = 0$).
Note that this scheme is amenable to analog (RF) beamforming as it can be implemented
with analog phase shifters alone. As a result, this scheme is of low-complexity
and is advantageous in mmW MIMO systems. In contrast, ${\bf f}_{\sf opt}$ requires
digital beamforming (in general) --- a higher complexity implementation --- as
it requires both phase shifters and gain control stages. Thus, it is important to
characterize the performance achievable with beamforming along the dominant path
in benchmarking the performance of the optimal scheme.

It is straightforward to see that 
this scheme results in the following received ${\sf SNR}$:
\begin{eqnarray}
\widetilde{ {\sf SNR}}_{\sf rx} & = & \frac{1}{L} \cdot
\max \Big( |\alpha_1|^2 + |\alpha_2|^2 |{\bf v}_1^{H} {\bf v}_2|^2,
{\hspace{0.03in}}
|\alpha_2|^2 + |\alpha_1|^2 |{\bf v}_1^{H} {\bf v}_2|^2 \Big)
\nonumber \\
& & {\hspace{0.05in}}
+ \frac{1}{L} \cdot
2 |\alpha_1| |\alpha_2| \cdot |{\bf v}_1^{H} {\bf v}_2| \cdot
|{\bf u}_1^{H} {\bf u}_2| \cdot
\cos \left(
\nu \right).
\nonumber
\end{eqnarray}
With $\max( |\alpha_1|^2, \hsppp |\alpha_2|^2) = K^2$ and
$\min( |\alpha_1|^2, \hsppp |\alpha_2|^2) = 1$ where $K \geq 1$, we have
\begin{eqnarray}
\widetilde{ {\sf SNR}}_{\sf rx} & = & \frac{1}{L} \cdot
\left( K^2 + |{\bf v}_1^{H} {\bf v}_2|^2 + 2 K \cdot
|{\bf v}_1^{H} {\bf v}_2| \cdot |{\bf u}_1^{H} {\bf u}_2| \cdot
\cos \left(
\nu \right) \right)
\nonumber \\
& \stackrel{(a)}{\leq} & \frac{1}{L} \cdot
\left( K^2 + |{\bf v}_1^{H} {\bf v}_2|^2 + 2 K \cdot
|{\bf v}_1^{H} {\bf v}_2| \cdot |{\bf u}_1^{H} {\bf u}_2| \right)
\nonumber
\end{eqnarray}
with equality achieved in (a) in the most optimistic scenario of coherent phase alignment
($\nu = 0$). Clearly, the upper bound is increasing in $K$, $|{\bf v}_1^{H} {\bf v}_2|$
and $|{\bf u}_1^{H} {\bf u}_2|$. Under favorable channel conditions
$( \{ |{\bf v}_1^{H} {\bf v}_2|, \hsppp |{\bf u}_1^{H} {\bf u}_2| \} \approx 1)$,
beamforming along a single path can yield $\frac{ (K + 1)^2 }{L}$,
corresponding to a case where the two paths coherently add at the receiver to increase
the signal amplitude. 
When only one of the paths is strong $(K \gg 1)$ or when ${\bf v}_1$ and ${\bf v}_2$
are electrically orthogonal, this coherent gain is lost and the beamforming gain is
$\frac{K^2}{L}$.

\subsection{${\bf v}_1$ and ${\bf v}_2$ are orthogonal}
\label{sec_appB4}
\begin{prop}
\label{prop_v1v2orth}
When ${\bf v}_1$ and ${\bf v}_2$ are electrically orthogonal, 
the non-unit-norm version of ${\bf f}_{\sf RSV}$ is given as
\begin{eqnarray}
\boxed{
{\bf f}_{\sf RSV} = \beta_{\sf opt} {\bf v}_1 +
e^{j \left(\angle{\alpha_1} - \angle{\alpha_2} - \angle{ {\bf u}_1^H {\bf u}_2 } \right) }
\sqrt{1 - \beta_{\sf opt}^2} {\bf v}_2 }
\nonumber
\end{eqnarray}
where
\begin{eqnarray}
\beta_{\sf opt}^2  = 
\frac{1}{ 2} \cdot \left[ 1 +
\frac{ |\alpha_1|^2 - |\alpha_2|^2}
{ \sqrt{  \left( |\alpha_1|^2 - |\alpha_2|^2 \right)^2 + 4
|\alpha_1|^2 |\alpha_2|^2 \cdot |{\bf u}_1^H {\bf u}_2|^2 }}
\right] .
\label{eq_betaopt}
\end{eqnarray}
The non-unit-norm version of ${\bf g}_{\sf opt}$ satisfies
\begin{eqnarray}
{\bf g}_{\sf opt}  = \alpha_1 \beta_{\sf opt} \cdot {\bf u}_1 +
e^{j \left(\angle{\alpha_1} - \angle{\alpha_2} - \angle{ {\bf u}_1^H {\bf u}_2 } \right) }
\alpha_2 \sqrt{1 - \beta_{\sf opt}^2}  \cdot {\bf u}_2.
\end{eqnarray}
\end{prop}
\begin{proof}
See Appendix~\ref{app_prop_v1v2orth}.
\end{proof}

\begin{figure*}[htb!]
\begin{center}
\begin{tabular}{cc}
\includegraphics[height=2.5in,width=3.0in] {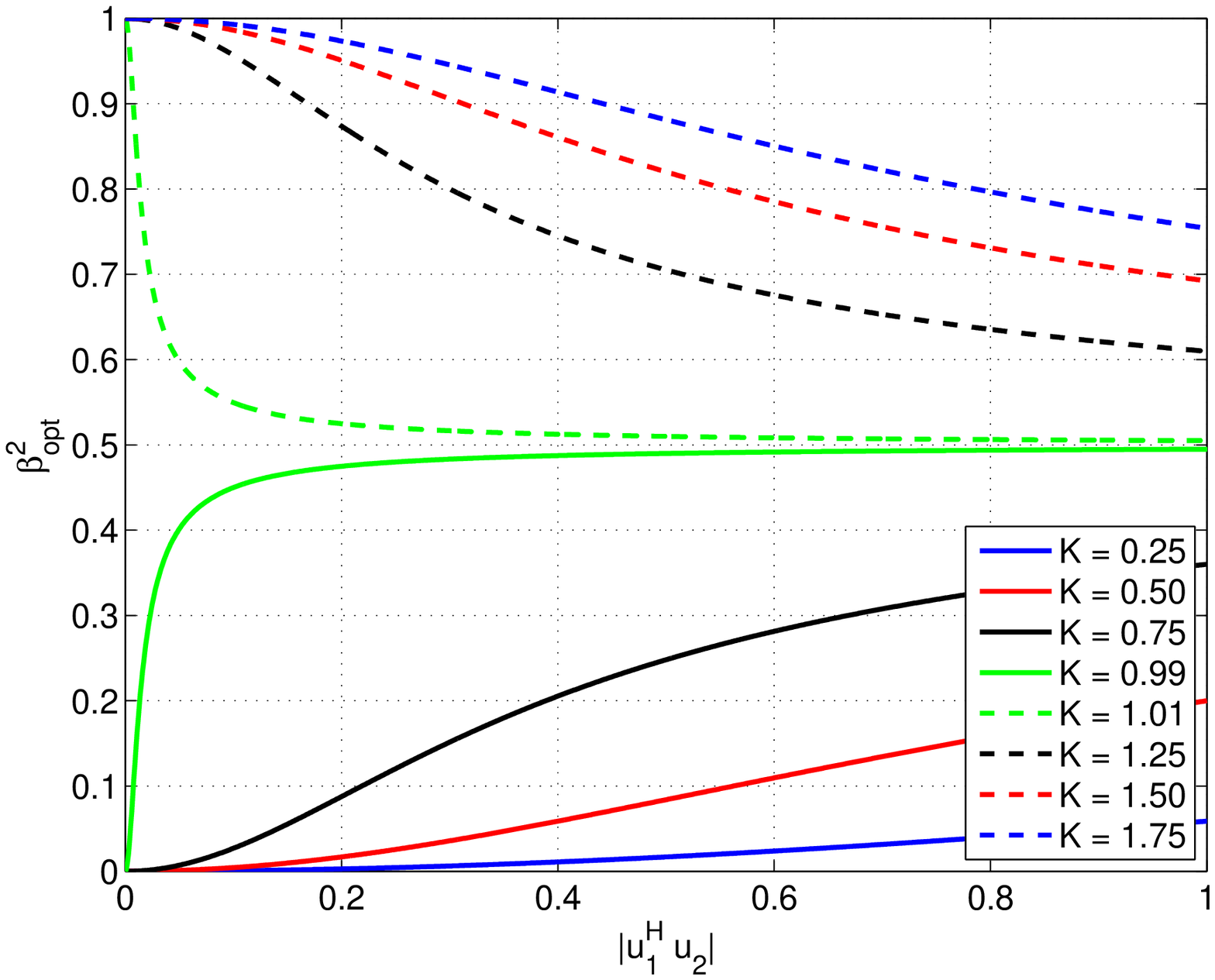}
&
\includegraphics[height=2.5in,width=3.0in] {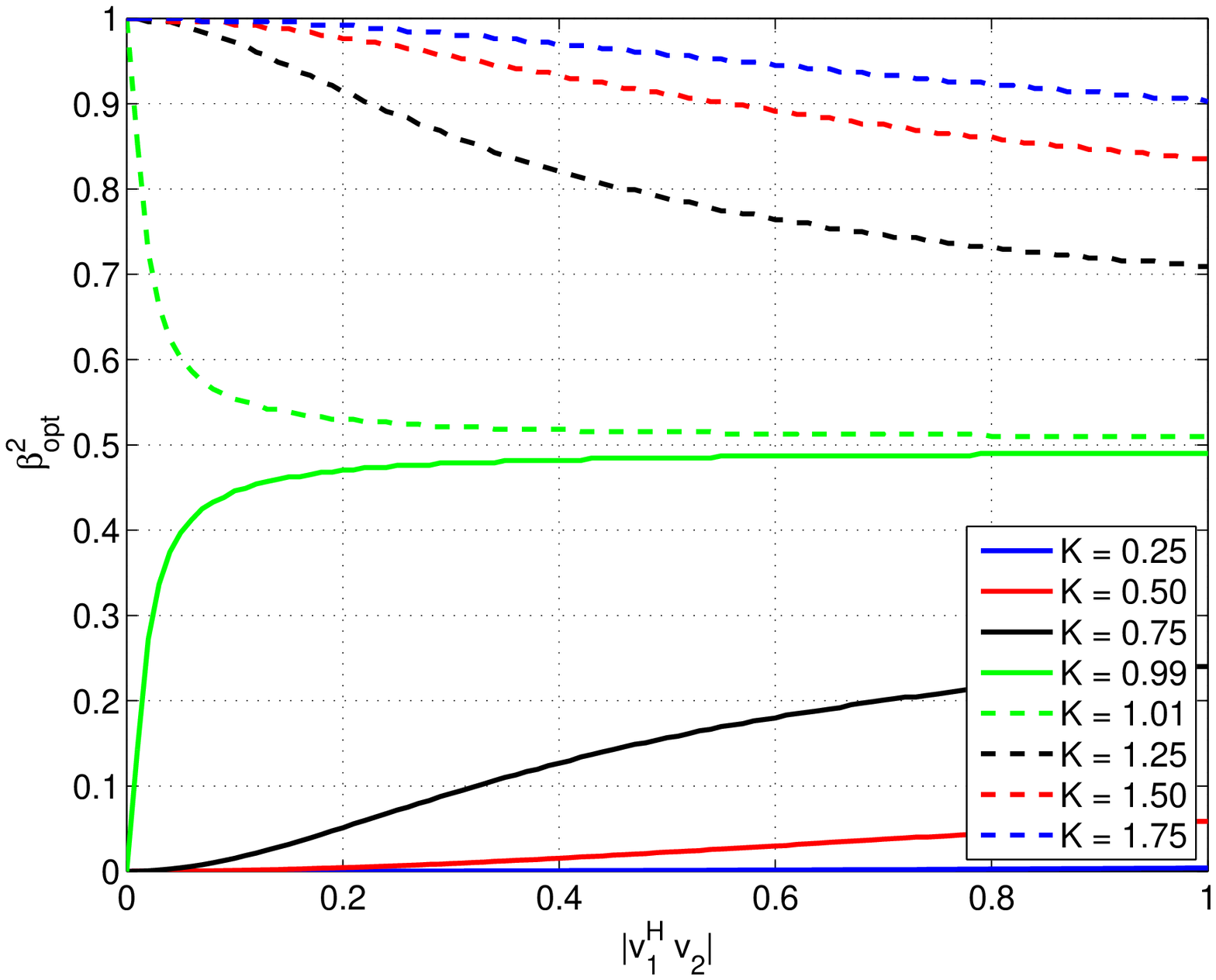}
\\
(a) & (b)
\end{tabular}
\caption{\label{fig1a} $\beta_{\sf opt}^2$ as a function of
$K = \frac{ |\alpha_1| } { |\alpha_2| }$ for different choices of:
a) ${\bf u}_1$ and ${\bf u}_2$ when ${\bf v}_1$ and ${\bf v}_2$ are orthogonal,
and b) ${\bf v}_1$ and ${\bf v}_2$ when ${\bf u}_1$ and ${\bf u}_2$ are
orthogonal.}
\end{center}
\end{figure*}

While the structure of $\beta_{\sf opt}^2$ is hard to visualize in general,
Fig.~\ref{fig1a}(a) plots it as a function of $|{\bf u}_1^{H} {\bf u}_2|$ for different
choices of $K =  \frac{ |\alpha_1| } { | \alpha_2 | }$. From Fig.~\ref{fig1a}(a), if ${\bf u}_1$
and ${\bf u}_2$ are orthogonal, we see that $\beta_{\sf opt}$ is either $1$ or $0$ with full
power allocated to the strongest path. In addition, a straightforward calculation shows that
\begin{eqnarray}
|{\bf u}_1^H {\bf u}_2| \rightarrow 1 \Longrightarrow \beta_{\sf opt}^2 \rightarrow \frac{ |\alpha_1|^2}
{ |\alpha_1|^2 + |\alpha_2|^2}.
\nonumber
\end{eqnarray}
In terms of loss with respect to beamforming along the dominant path, a simple calculation
shows that
\begin{eqnarray}
\Delta {\sf SNR} \triangleq
\frac{ \widetilde{\sf SNR}_{\sf rx} \Big|_{\sf RSV} }
{ \widetilde{\sf SNR}_{\sf rx} \Big|_{\sf Dom. \hsppp path} } =
\frac{ |\alpha_1|^2 + |\alpha_2|^2 +
\sqrt{ |\alpha_1|^4 + |\alpha_2|^4 + 2 |\alpha_1|^2 |\alpha_2|^2 \cdot
\left(2 |{\bf u}_1^H {\bf u}_2|^2 - 1 \right) }}
{ 2 \cdot \max \left( |\alpha_1|^2, \hsppp |\alpha_2|^2 \right) }.
\label{eq_DeltaSNR_v1v2orth}
\end{eqnarray}
Clearly, $\Delta {\sf SNR}$ is increasing in $|{\bf u}_1^H {\bf u}_2|$ with
\begin{eqnarray}
1 \leq \Delta {\sf SNR} \leq 1 + \frac{ \min \left( |\alpha_1|^2, \hsppp |\alpha_2|^2 \right) }
{ \max \left( |\alpha_1|^2, \hsppp |\alpha_2|^2 \right)}
\nonumber
\end{eqnarray}
where the lower bound is realized when ${\bf u}_1$ and ${\bf u}_2$ are orthogonal and
the upper bound is realized when they are parallel. The above relationship clearly
shows that the worst-case performance loss with beamforming along a single path is $3$ dB.
This SNR loss (in dB) is plotted in Fig.~\ref{fig2a}(a) as a function of $K =
\frac{ |\alpha_1|}{ |\alpha_2|}$ for different choices of $|{\bf u}_1^H {\bf u}_2|$.

\begin{figure*}[htb!]
\begin{center}
\begin{tabular}{cc}
\includegraphics[height=2.5in,width=3.0in] {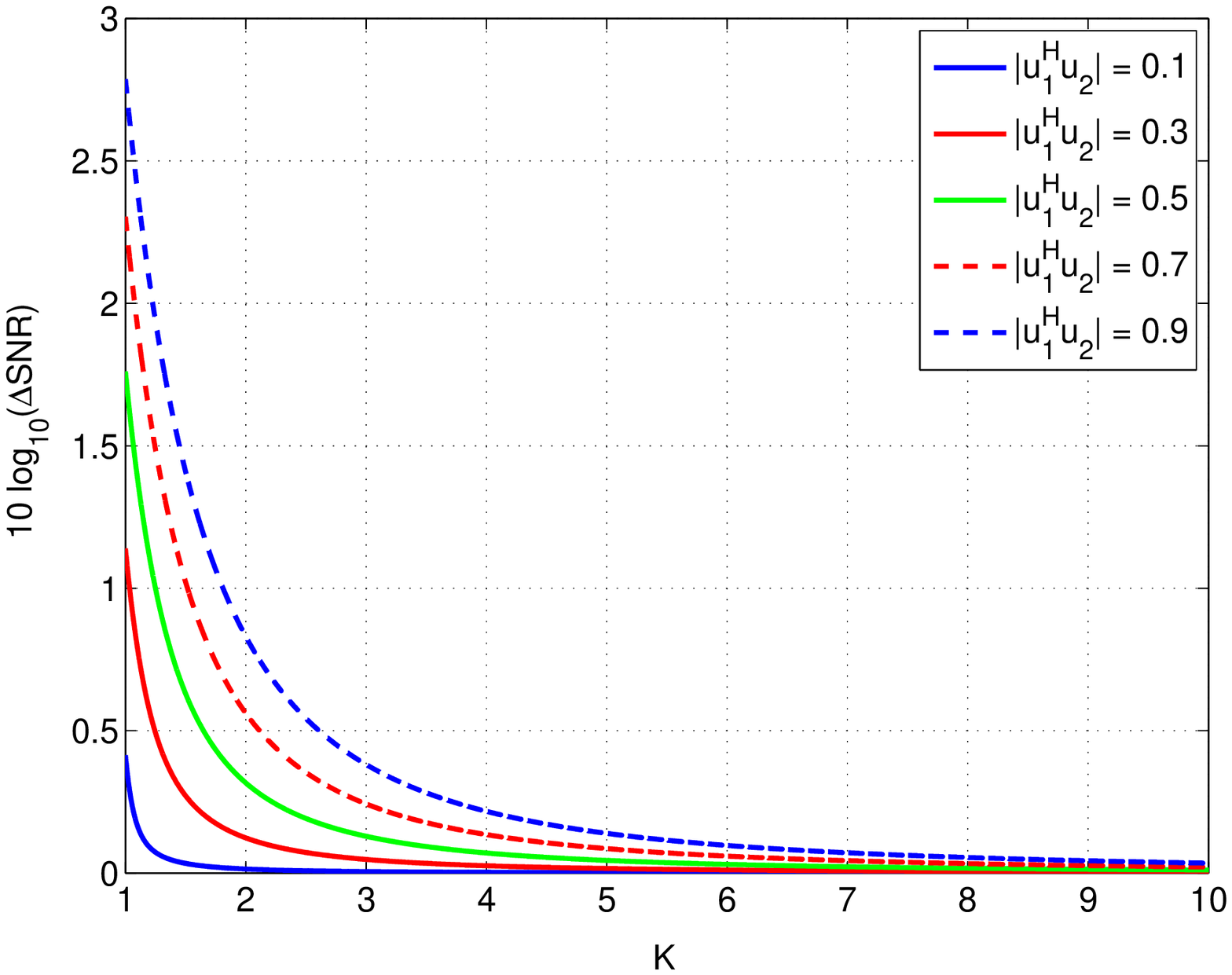}
&
\includegraphics[height=2.5in,width=3.0in] {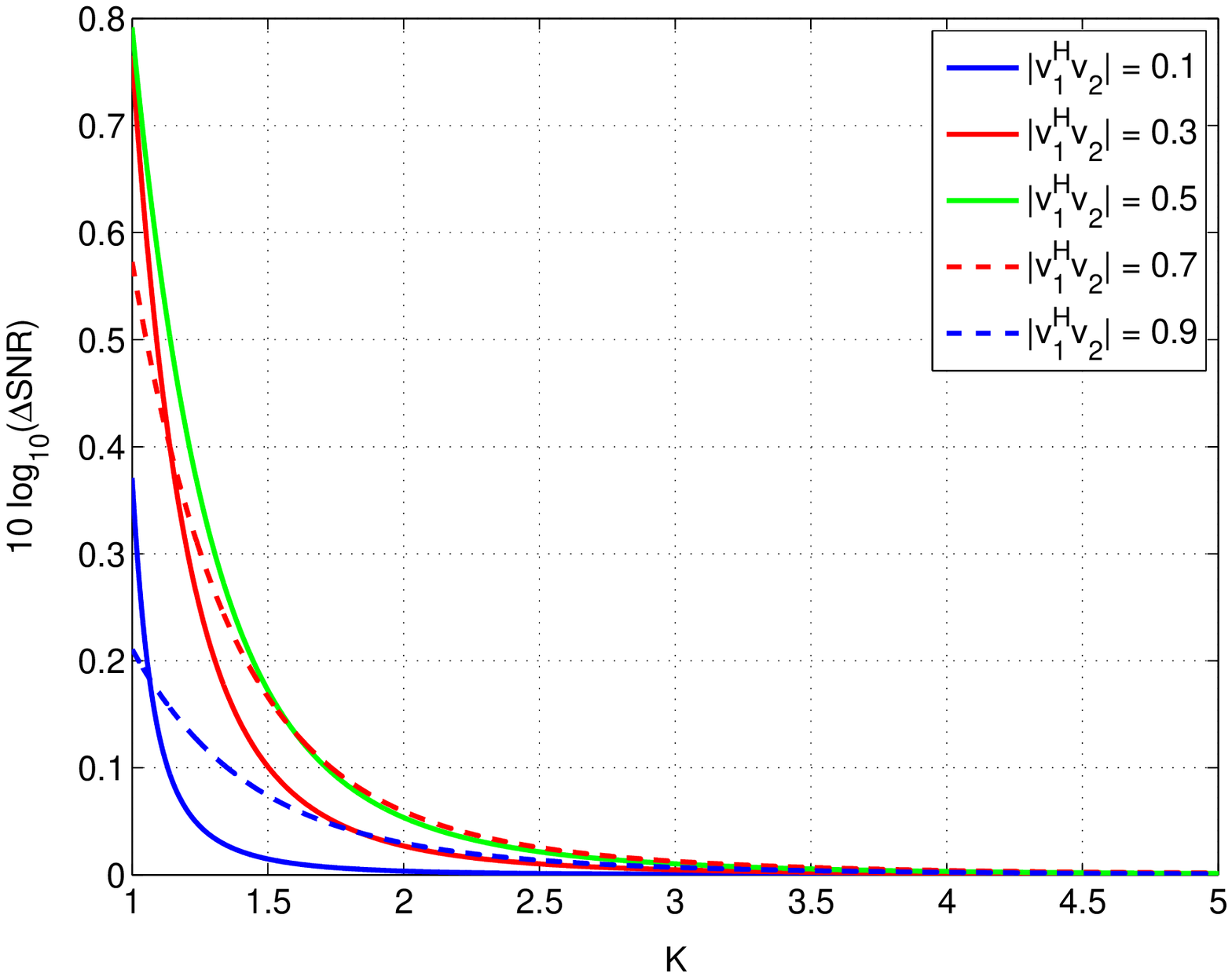}
\\
(a) & (b)
\end{tabular}
\caption{\label{fig2a} $\Delta{\sf SNR}$ between the optimal beamforming scheme and
beamforming along the strongest path as a function of $K$:
(a) as given by~(\ref{eq_DeltaSNR_v1v2orth}) when ${\bf v}_1$ and ${\bf v}_2$ are
orthogonal, and (b) as given by~(\ref{eq_DeltaSNR_u1u2orth}) when ${\bf u}_1$ and
${\bf u}_2$ are orthogonal.}
\end{center}
\end{figure*}

\subsection{${\bf u}_1$ and ${\bf u}_2$ are orthogonal}
\begin{prop}
\label{prop_u1u2orth}
If ${\bf u}_1$ and ${\bf u}_2$ are electrically orthogonal, 
the non-unit-norm version of ${\bf f}_{\sf RSV}$ is given as
\begin{eqnarray}
\boxed{
{\bf f}_{\sf RSV} = \beta_{\sf opt} {\bf v}_1 +
e^{ - j \angle{ {\bf v}_1^{H} {\bf v}_2}}
\sqrt{1 - \beta_{\sf opt}^2} {\bf v}_2 }
\nonumber
\end{eqnarray}
where
\begin{eqnarray}
\beta_{\sf opt}^2 & = & \left\{
\begin{array}{cc}
\frac{ {\cal A} + \sqrt{ {\cal B}} }{2 {\hspace{0.01in}} {\cal C}}
& {\rm if} {\hspace{0.1in}} |\alpha_1| \geq |\alpha_2| \\
\frac{ {\cal A} - \sqrt{ {\cal B}} }{2 {\hspace{0.01in}} {\cal C}}
& {\rm if} {\hspace{0.1in}} |\alpha_1| < |\alpha_2|
\end{array}
{\hspace{0.10in}} {\rm with} \right.
\nonumber \\
{\cal A} & = & \frac{ (|\alpha_1|^2 - |\alpha_2|^2)^2 }
{ |{\bf v}_1^{H} {\bf v}_2|^2 } + 2 |\alpha_1|^2 \cdot ( |\alpha_1|^2 + |\alpha_2|^2 )
\nonumber \\
{\cal B} & = & \frac{ (|\alpha_1|^2 - |\alpha_2|^2 )^4 }
{ |{\bf v}_1^{H} {\bf v}_2|^4 } +
\frac{ 4 |\alpha_1|^2 |\alpha_2|^2 }{ | {\bf v}_1^{H} {\bf v}_2|^2 }
\cdot \left( |\alpha_1|^2 - |\alpha_2|^2 \right)^2
\nonumber \\
{\cal C} & = & \left( 1 + \frac{1}{ |{\bf v}_1^{H} {\bf v}_2|^2} \right)
\cdot \left( |\alpha_1|^2 + |\alpha_2|^2  \right)^2 -
\frac{ 4 |\alpha_1|^2 |\alpha_2|^2 }{ |{\bf v}_1^{H} {\bf v}_2 |^2 }.
\nonumber
\end{eqnarray}
The non-unit-norm version of ${\bf g}_{\sf opt}$ satisfies
\begin{eqnarray}
{\bf g}_{\sf opt} =  \alpha_1 \left( \beta_{\sf opt} + | {\bf v}_1^H {\bf v}_2 |
 \sqrt{1 - \beta_{\sf opt}^2} \right) \cdot {\bf u}_1.
\end{eqnarray}
\end{prop}
\begin{proof}
See Appendix~\ref{app_prop_u1u2orth}.
\end{proof}

Fig.~\ref{fig1a}(b) plots $\beta_{\sf opt}^2$ as a function of $|{\bf v}_1^{H} {\bf v}_2|$
for different choices of $K =  \frac{ |\alpha_1| } { | \alpha_2 | }$. As before,
Fig.~\ref{fig1a}(b) shows that $\beta_{\sf opt}$ converges to $0$ or $1$ as ${\bf v}_1$
and ${\bf v}_2$ become more orthogonal. A straightforward calculation also shows that 
\begin{eqnarray}
|{\bf v}_1^{H} {\bf v}_2| \rightarrow 1 \Longrightarrow
\beta_{\sf opt}^2 \rightarrow \frac{ |\alpha_1|^4 }
{ |\alpha_1|^4 + |\alpha_2|^4 }.  \nonumber
\end{eqnarray}
In between these two extremes, we have
\begin{eqnarray}
\frac{1}{L} \cdot \max \left( |\alpha_1|^2, {\hspace{0.02in}} | \alpha_2|^2 \right)
\leq \widetilde{ {\sf SNR} }_{\sf rx} \leq \frac{ |\alpha_1|^2 + |\alpha_2|^2}{L}.
\nonumber
\end{eqnarray}

In terms of loss with respect to beamforming along the dominant path, a simple
calculation shows that
\begin{eqnarray}
\Delta {\sf SNR} = \frac{ \widetilde{\sf SNR}_{\sf rx} \Big|_{\sf RSV} }
{ \widetilde{\sf SNR}_{\sf rx} \Big|_{\sf Dom. \hsppp path} } =
\frac{  |\alpha_1|^2 + |\alpha_2|^2
- \left( 1 - |{\bf v}_1^{H} {\bf v}_2|^2  \right) \cdot
\left( \frac{ \beta_{\sf opt}^2 \cdot |\alpha_2|^2 + (1 - \beta_{\sf opt}^2) \cdot |\alpha_1|^2 }
{  1 + 2 \beta_{\sf opt} \sqrt{1 - \beta_{\sf opt}^2} \cdot |{\bf v}_1^{H} {\bf v}_2|  }
\right)
}{
\max\left( |\alpha_1|^2 + |\alpha_2|^2 |{\bf v}_1^H {\bf v}_2|^2, \hspp
|\alpha_2|^2 + |\alpha_1|^2 |{\bf v}_1^H {\bf v}_2|^2 \right) }
\label{eq_DeltaSNR_u1u2orth}
\end{eqnarray}
where $\beta_{\sf opt}$ is as in the statement of the proposition. 
While this expression is also hard to visualize, Fig.~\ref{fig2a}(b) plots it as a function
of $K = \frac{ |\alpha_1|} { |\alpha_2|} $ for different values of
$|{\bf v}_1^H{\bf v}_2|$. 
With $K = \frac{ |\alpha_1|} { |\alpha_2|} \geq 1$, note that $\Delta {\sf SNR}$ can be rewritten as
\begin{eqnarray}
\Delta {\sf SNR} = 1 + \left[ \frac{ \left( 1 - |{\bf v}_1^H {\bf v}_2|^2 \right) \cdot
\sqrt{1 - \beta_{\sf opt}^2} }
{ 1  + 2 \beta_{\sf opt} \sqrt{1 - \beta_{\sf opt}^2 } \cdot |{\bf v}_1^H {\bf v}_2|} \right]
\cdot
\left[ \frac{ \sqrt{1 - \beta_{\sf opt}^2} \cdot (1 - K^2) +
2 \beta_{\sf opt} | {\bf v}_1^H {\bf v}_2| }
{ K^2 + |{\bf v}_1^H {\bf v}_2|^2 } \right].
\nonumber
\end{eqnarray}
While optimizing the above expression in terms of $K$ is difficult given the complicated functional
involvement of $K$ in the above expression, by treating $\beta_{\sf opt}$ as a fixed quantity,
it is straightforward to see that the above expression is decreasing in $K$. Without being rigorous,
this argument suggests that the above expression is maximized at $K = 1$. Substituting $K = 1$, we
have $\beta_{\sf opt}^2 = \frac{1}{2}$ and
\begin{eqnarray}
\Delta {\sf SNR} = \frac{ 1 + |{\bf v}_1^H {\bf v}_2|} { 1 + |{\bf v}_1^H {\bf v}_2|^2 }
. \nonumber
\end{eqnarray}
It is easy to see that the above expression is maximized at $|{\bf v}_1^H {\bf v}_2|
= \sqrt{2} - 1$ with a maximum value of $\Delta {\sf SNR} = \frac{\sqrt{2} + 1 }{2} =
0.8175$ dB. Thus, beamforming along the dominant path is no worser than $0.8175$ dB in
terms of optimal beamforming performance. This trend is reinforced by the $\Delta {\sf SNR}$
plot in Fig.~\ref{fig2a}(b) as a function of $K = \frac{ |\alpha_1| } { |\alpha_2| }$ for
different values of $|{\bf v}_1^H {\bf v}_2|$.

\subsection{${\bf v}_1$ and ${\bf v}_2$ are parallel}
If ${\bf v}_1$ and ${\bf v}_2$ are parallel (or nearly parallel), we can use
$|{\bf v}_1^H {\bf v}_2| \approx 1$ to rewrite ${\widetilde{\sf SNR}}_{\sf rx}$ as
\begin{eqnarray}
\widetilde{\sf SNR}_{\sf rx} & = &
\frac{
\begin{split}
& |\alpha_1|^2 + |\alpha_2|^2
+ 2 |\alpha_1| | \alpha_2 | \cdot | {\bf u}_1^{H} {\bf u}_2| \cdot \cos \left( \nu \right)
+ 2 \beta \sqrt{1 - \beta^2} \cdot \cos(\phi) \cdot
\nonumber \\
& {\hspace{1.0in}}
\left( |\alpha_1|^2 + |\alpha_2|^2 + 2 |\alpha_1| |\alpha_2| \cdot | {\bf u}_1^{H} {\bf u}_2 |
\cdot \cos(\nu) \right)
\end{split}
}
{L \cdot \left( 1 + 2 \beta \sqrt{1 - \beta^2}  \cdot \cos( \phi) \right) }
\nonumber \\
& = & \frac{1}{L} \cdot
\left(
|\alpha_1|^2 + |\alpha_2|^2 + 2 |\alpha_1| |\alpha_2| \cdot | {\bf u}_1^{H} {\bf u}_2 | \cdot \cos(\nu)
\right).
\nonumber
\end{eqnarray}
Clearly, the above objective function is independent of $\beta$ and $\theta$. Therefore,
any power allocation scheme across the two paths achieves the above gain. A corollary of
this observation is that beamforming along the dominant path is as good as the optimal
beamforming scheme ($\Delta {\sf SNR} = 1$).

\subsection{${\bf u}_1$ and ${\bf u}_2$ are parallel}
\begin{prop}
\label{prop_u1u2par}
If ${\bf u}_1$ and ${\bf u}_2$ are parallel, 
the non-unit-norm version of ${\bf f}_{\sf RSV}$ is given as
\begin{eqnarray}
\boxed{
{\bf f}_{\sf RSV} = 
\beta_{\sf opt} {\bf v}_1 +
e^{j \left( \angle{ \alpha_1} - \angle{ \alpha_2 } - \angle{ {\bf u}_1^{H} {\bf u}_2} \right) }
\sqrt{1 - \beta_{\sf opt}^2} {\bf v}_2
} 
\nonumber
\end{eqnarray}
where
\begin{eqnarray}
\beta_{\sf opt}^2  = \frac{|\alpha_1|^2}{|\alpha_1|^2 + |\alpha_2|^2} .
\nonumber
\end{eqnarray}
The non-unit-norm version of ${\bf g}_{\sf opt}$ follows from expanding out
${\sf H} \hsppp {\bf f}_{\sf opt}$ and is not provided here.
\end{prop}
\begin{proof}
See Appendix~\ref{app_prop_u1u2par}.
\end{proof}

Note that $\beta_{\sf opt}^2$ mimics a maximum ratio combining solution, allocating power
to each path in proportion to the gain of that path. With $K = \frac{ |\alpha_1|} {|\alpha_2| }
\geq 1$, the ${\sf SNR}$ loss can be written as
\begin{eqnarray}
\Delta {\sf SNR} & = & \frac{ \widetilde{\sf SNR}_{\sf rx} \Big|_{\sf RSV} }
{ \widetilde{\sf SNR}_{\sf rx} \Big|_{\sf Dom. \hsppp path} }
\nonumber \\
& = & \frac{ |\alpha_1|^2 + |\alpha_2|^2 +
2 |\alpha_1| |\alpha_2| \cos(\nu)\cdot |{\bf v}_1^H {\bf v}_2|  }
{ \max \left( |\alpha_1|^2 + |{\bf v}_1^H {\bf v}_2|^2 |\alpha_2|^2,
\hsppp |\alpha_2|^2 + |{\bf v}_1^H {\bf v}_2|^2 |\alpha_1|^2 \right)
+ 2 |\alpha_1| |\alpha_2| |{\bf v}_1^H {\bf v}_2| \cos(\nu) }
\nonumber \\
& = & 1+ \frac{ |\alpha_2|^2 \cdot \left(1 - |{\bf v}_1^H {\bf v}_2|^2 \right) }
{  |\alpha_1|^2 + |{\bf v}_1^H {\bf v}_2|^2 |\alpha_2|^2
+ 2 |\alpha_1| |\alpha_2| |{\bf v}_1^H {\bf v}_2| \cdot \cos(\nu) }
\nonumber
\end{eqnarray}
This ${\sf SNR}$ loss term is plotted in Fig.~\ref{fig3a} as a function of $K =
\frac{ |\alpha_1| }{ |\alpha_2| }$ for different choices of $| {\bf v}_1^H {\bf v}_2|$ and
$\nu$. From this study, we see that $\Delta {\sf SNR}$ can be significantly larger than
$3$ dB provided that both paths are approximately similar in terms of gain and are also
essentially parallel, but with opposite phases (characterized by $\nu = 180^o$). In this
setting, the right singular vector combines the gains in both paths by appropriate phase
compensation. On the other hand, beamforming along only the strongest path leads to
destructive interference of the signal from the sub-dominant path resulting in significant
performance loss. Barring these extreme conditions, this study also shows that the
performance loss is similar to the $3$ dB characterization in other settings.

\begin{figure*}[!]
\begin{center}
\begin{tabular}{cc}
\includegraphics[height=2.4in,width=2.8in] {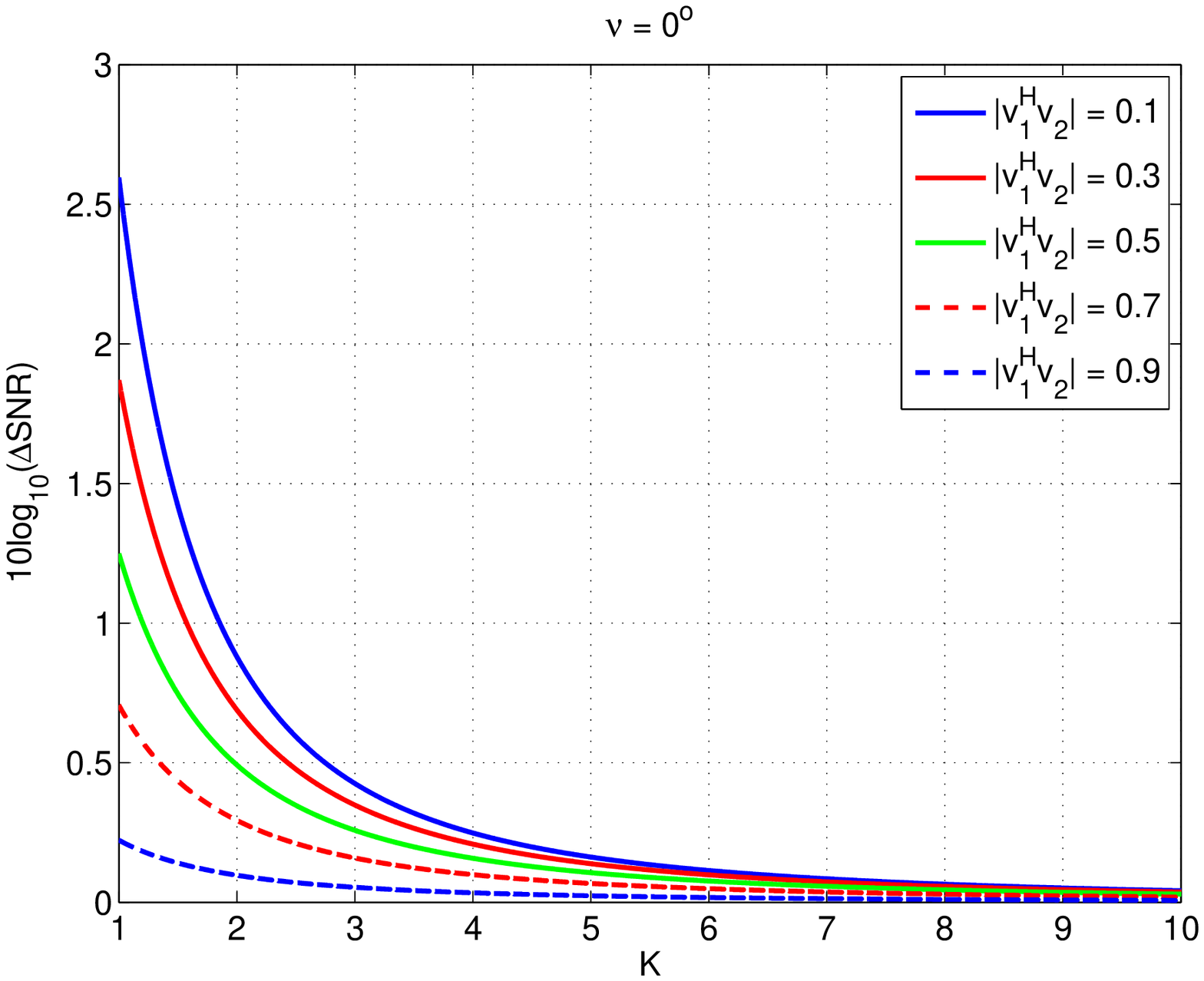}
&
\includegraphics[height=2.4in,width=2.8in] {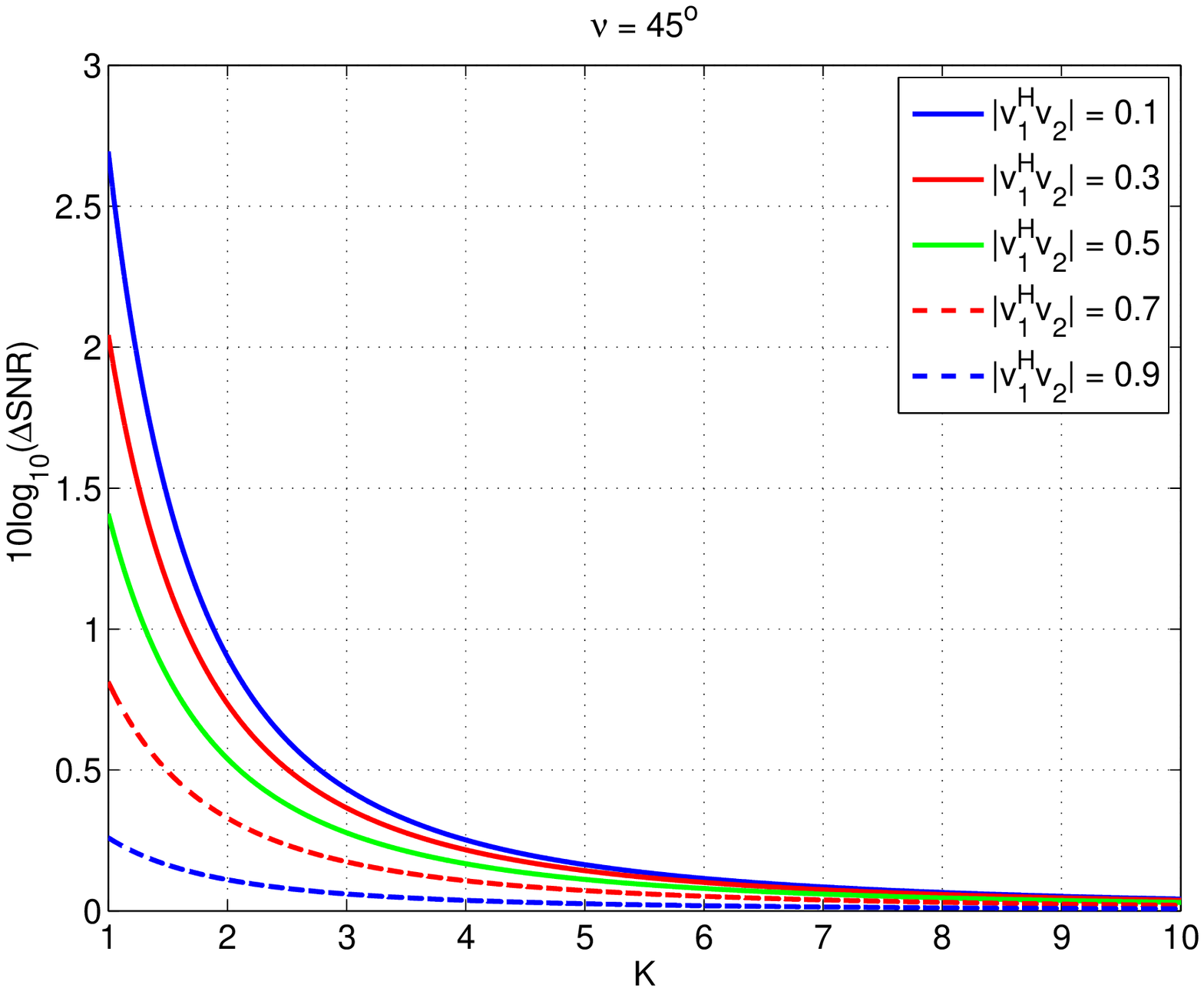}
\\
(a) & (b)
\\
\includegraphics[height=2.4in,width=2.8in] {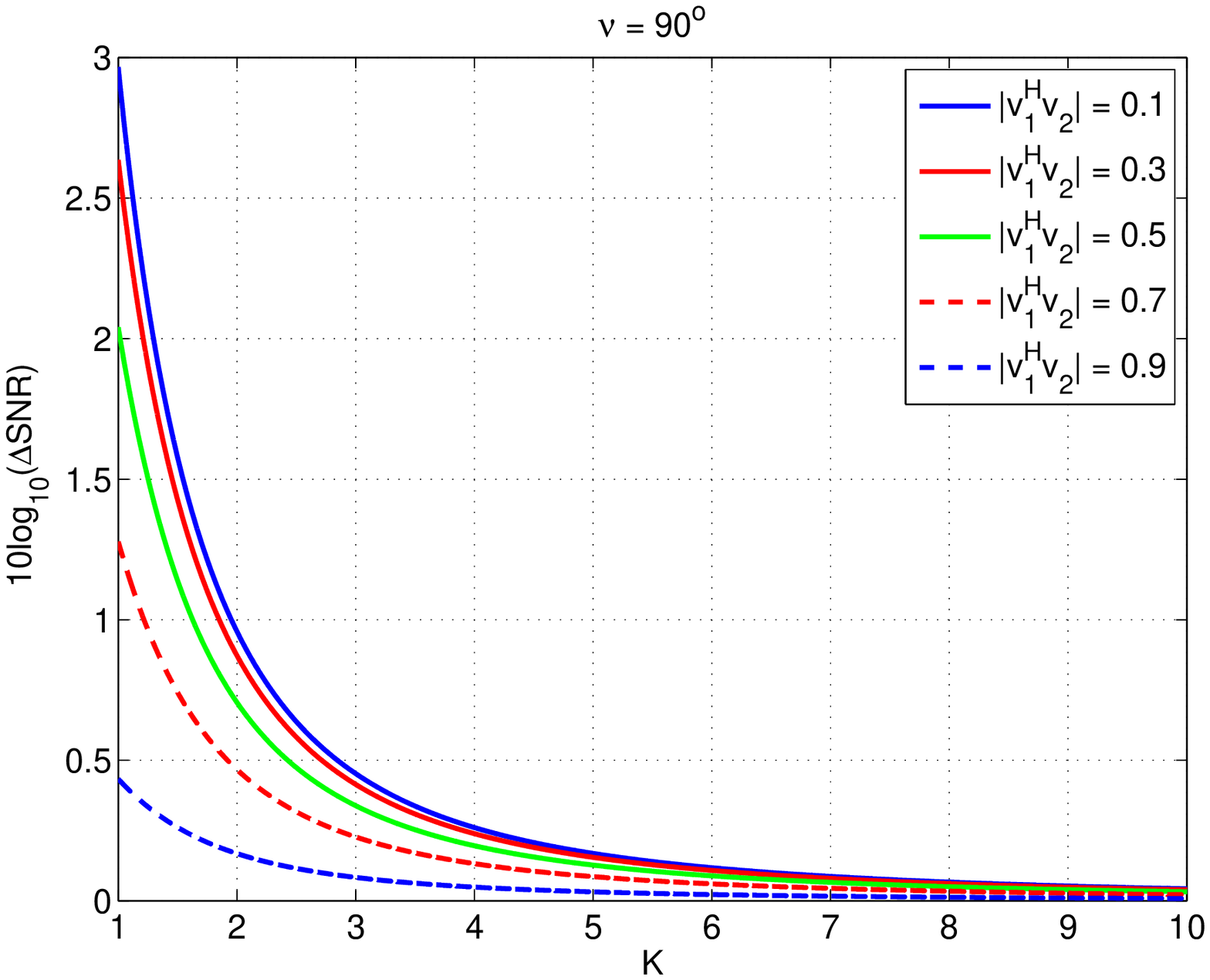}
&
\includegraphics[height=2.4in,width=2.8in] {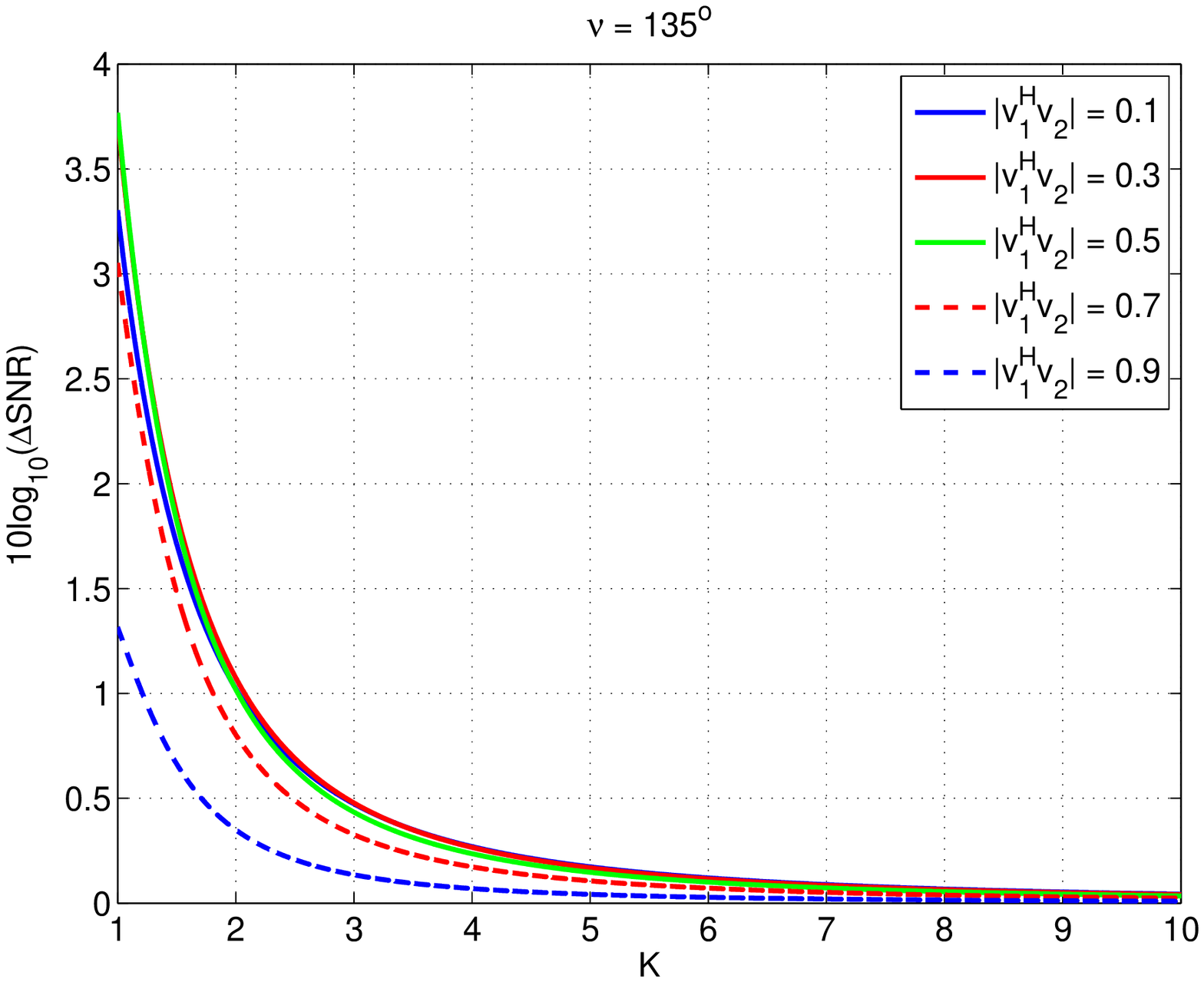}
\\
(c) & (d)
\\
\includegraphics[height=2.4in,width=2.8in] {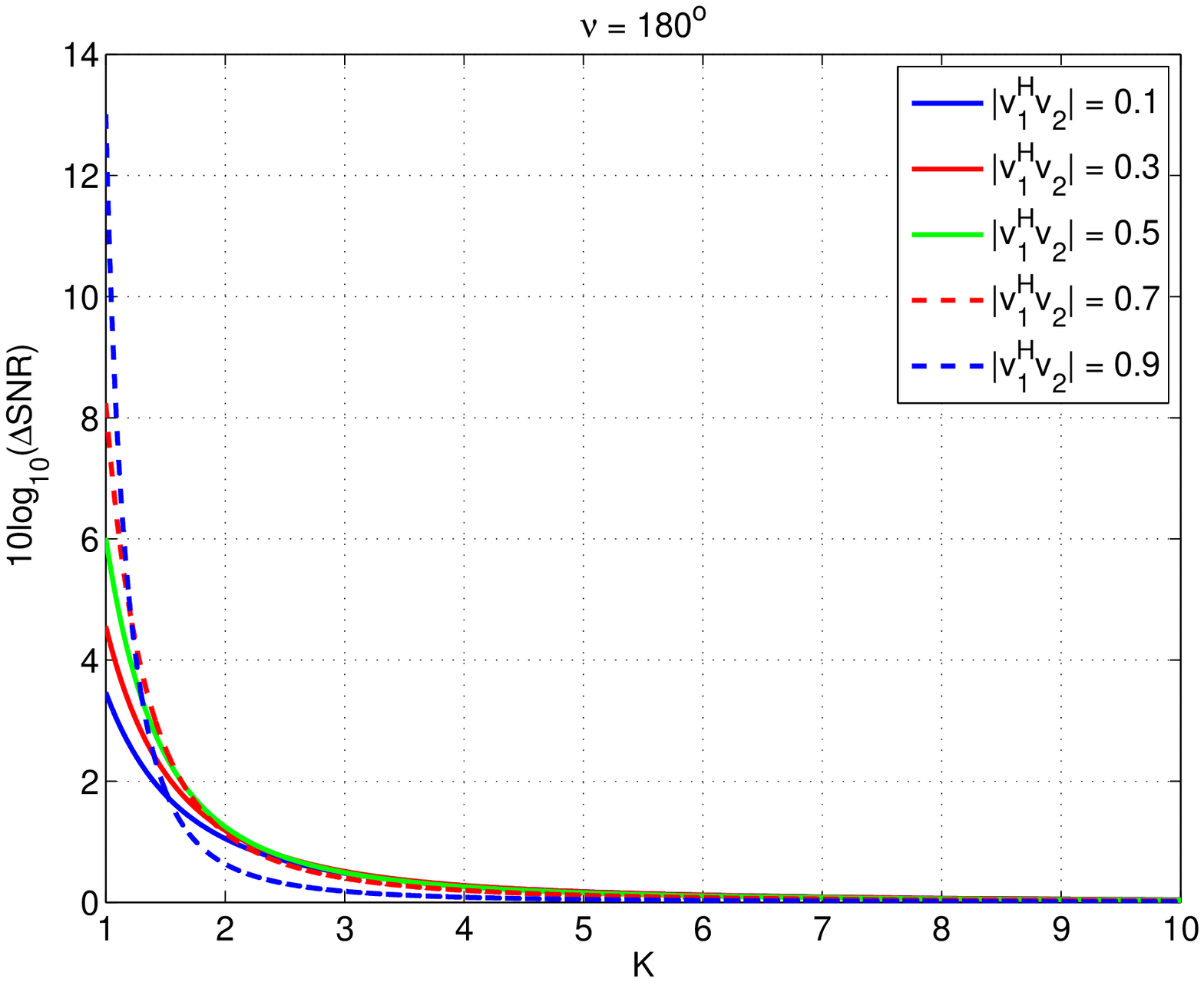}
&
\includegraphics[height=2.4in,width=2.8in] {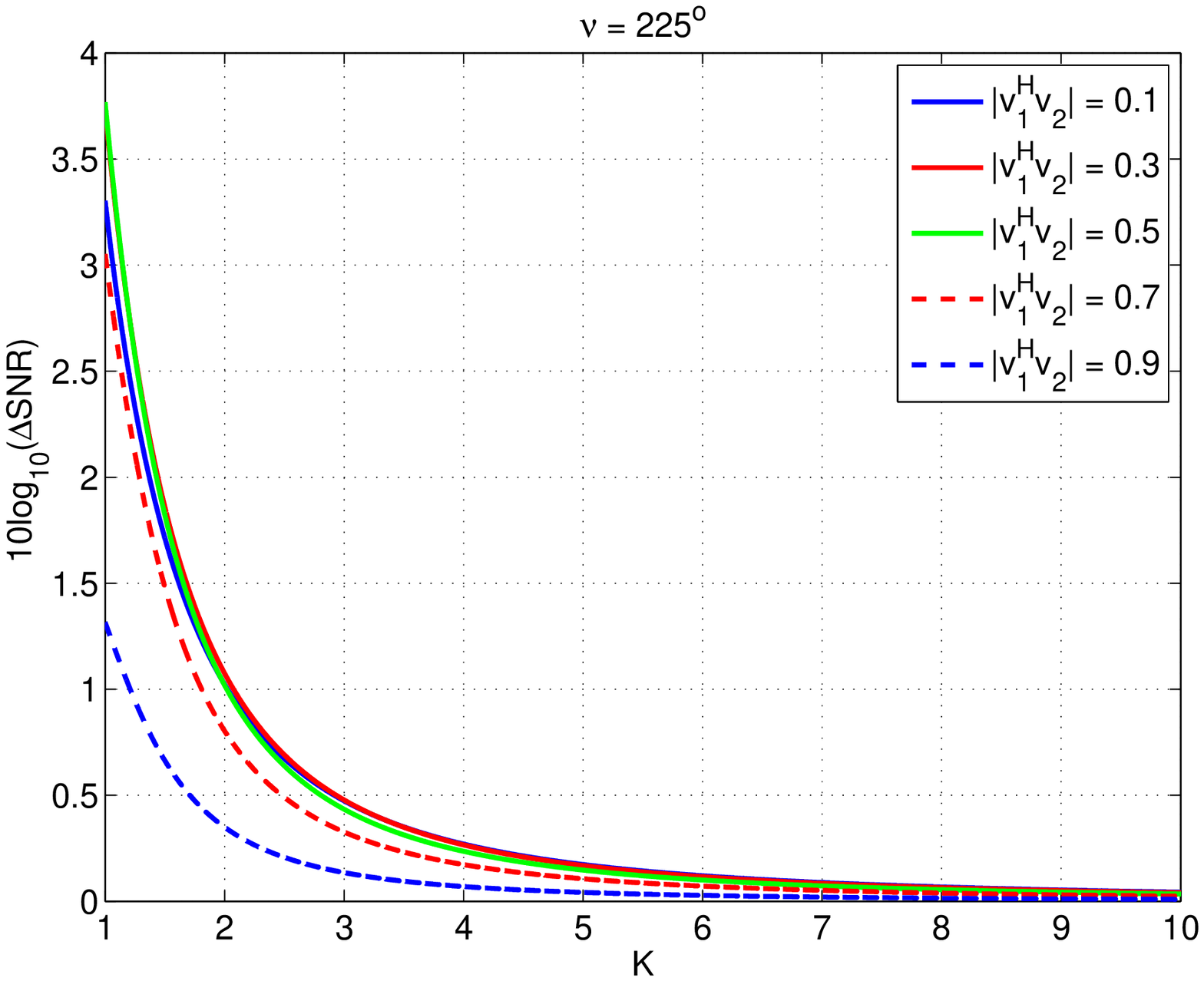}
\\
(e) & (f)
\end{tabular}
\caption{\label{fig3a}
$\Delta{\sf SNR}$ between the optimal beamforming scheme and beamforming along the strongest
path as a function of $K$ when ${\bf u}_1$ and ${\bf u}_2$ are parallel for different choices
of $\nu$: (a) $\nu = 0^o$, (b) $\nu = 45^o$, (c) $\nu = 90^o$, (d) $\nu = 135^o$,
(e) $\nu = 180^o$, and (f) $\nu = 225^o$.}
\end{center}
\end{figure*}

\subsection{Beamforming with equal power allocation}
Another simple scheme allocates power equally to both the directions
$(\beta = \frac{1}{\sqrt{2}})$. Note that this scheme requires a digital
beamformer (in general) since the sum of two CPO beams does not have a constant
amplitude. For this scheme, it is straightforward to see that
\begin{eqnarray}
\widetilde{ {\sf SNR}} _{\sf rx} & = & 
\frac{ {\cal A}_1 + {\cal A}_2 }
{ L \cdot 2 \left( 1 + |{\bf v}_1^{H} {\bf v}_2| \cdot
\cos( \phi 
) \right)}
\nonumber \\
{\cal A}_1 & = & \left( |\alpha_1|^2 + |\alpha_2|^2 \right)
\cdot \Big[ 1 + |{\bf v}_1^{H} {\bf v}_2|^2 +
2 |{\bf v}_1^{H} {\bf v}_2| \cos( 
\phi ) \Big] \nonumber \\
{\cal A}_2 & = & 2 |\alpha_1| | \alpha_2| \cdot |{\bf u}_1^{H} {\bf u}_2 |
\cdot \Big[
|{\bf v}_1^{H} {\bf v}_2|^2 \cdot \cos \left( \nu + \phi
\right) 
+ 2 |{\bf v}_1^{H} {\bf v}_2| \cdot
\cos  \left( \nu
\right) + \cos \left( \nu - \phi
\right) \Big]. \nonumber
\end{eqnarray}
While the optimal choice of $\phi$ is unclear for this scheme, in the scenario of
coherent phase alignment ($\phi = \nu = 0$), with $\max( |\alpha_1|^2, \hsppp
|\alpha_2|^2) = K^2$ and $\min( |\alpha_1|^2, \hsppp |\alpha_2|^2) = 1$ where
$K \geq 1$, we have
\begin{eqnarray}
\widetilde{ {\sf SNR}}_{\sf rx} & = & \frac{
\left( |\alpha_1|^2 + |\alpha_2|^2 + 2 |\alpha_1| |\alpha_2| \cdot
| {\bf u}_1^{H} {\bf u}_2 |  \right) \cdot
\big( 1 +  | {\bf v}_1^{H} {\bf v}_2| \big)^2  }
{L \cdot 2 \big( 1 + |{\bf v}_1^{H} {\bf v}_2| \big) }
\nonumber \\
& = & \frac{ \big(1 + |{\bf v}_1^{H} {\bf v}_2 | \big) }{L \cdot 2} \cdot
\left( |\alpha_1|^2 + |\alpha_2|^2 + 2 |\alpha_1| |\alpha_2| \cdot
| {\bf u}_1^{H} {\bf u}_2 |  \right)
\nonumber \\
& = & \frac{ \big(1 + |{\bf v}_1^{H} {\bf v}_2 | \big) }{L \cdot 2} \cdot
\left( K^2 + 1 + 2 K \cdot | {\bf u}_1^{H} {\bf u}_2 |  \right)
\nonumber
\end{eqnarray}
Under favorable channel conditions $( \{ |{\bf v}_1^{H} {\bf v}_2|, \hsppp
|{\bf u}_1^{H} {\bf u}_2| \} \approx 1)$, equal power beamforming can add
signals coherently to yield $\frac{ (K + 1)^2 }{L}$, whereas when $K \gg 1$, we
have a gain of $\frac{ K^2}{L} \cdot \frac{ 1 + |{\bf v}_1^{H} {\bf v}_2 |  }{2}$.
If ${\bf v}_1$ and ${\bf v}_2$ ar electrically orthogonal, it is clear that half
the power (along ${\bf v}_2$) is wasted resulting in a $3$ dB loss over the
scheme where the entire power is directed along the dominant path (${\bf v}_1$).

\section{Directional Beamforming at Both Ends}
While we have so far considered the case of directional beamforming at the
transmitter, the receiver uses a matched filter corresponding to such a scheme,
which may not be directional. We now consider the case of directional
beamforming at both ends. In Fig.~\ref{fig4}, we plot the complementary
cumulative distribution function (CCDF) of the loss in ${\sf SNR}_{\sf rx}$
with such a bi-directional scheme relative to the optimal beamforming scheme
for different choices of $L$. The gains of the paths as well as their
directions are chosen independently and identically distributed (i.i.d.) from
a certain path loss model and over the $120^{\sf o}$ field-of-view of the
arrays. From this figure, we note that for a large fraction of the channel
realizations, beamforming along the dominant direction only results in a small
performance loss. In particular, the median losses in the three cases ($L = 2$, $3$
and $5$) are $0.3$ dB, $0.95$ dB and $1.85$ dB, and the $90$-th percentile losses
are $1.8$ dB, $2.45$ dB and $3.4$ dB. Thus, this study suggests that directional
beamforming could serve as a useful low-complexity scheme with good performance
in the mmW regime.

\begin{figure*}[!]
\begin{center}
\includegraphics[height=2.8in,width=3.5in] {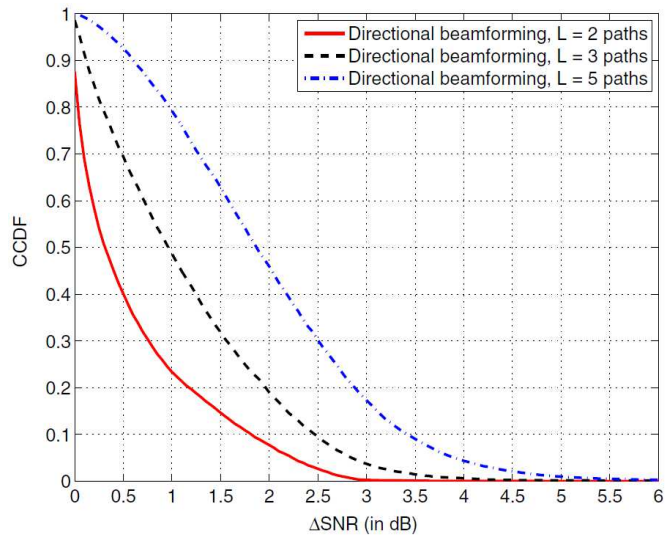}
\end{center}
\caption{\label{fig4} Complementary CDF of $\Delta {\sf SNR}$ as a function
of $L$.}
\end{figure*}

\section{Concluding Remarks}
This paper developed an explicit mapping and dependence of the optimal beamformer
structure on the different aspects of the sparse channel that characterize
propagation in the mmW regime. This study showed that the optimal beamformer approaches
dominant path (directional) beamforming as either the AoDs or AoAs of the paths
become more (electrically) orthogonal. In general, if the AoDs or AoAs are not
orthogonal, optimal beamforming entails appropriate power allocation and phase
compensation across the paths. While specific channel realizations can be
constructed to ensure that directional beamforming can suffer significantly
relative to the optimal scheme, in a distributional sense, the loss in received
${\sf SNR}$ is expected to be minimal. Furthermore, this small additional gain
in received ${\sf SNR}$ with optimal beamforming comes at the cost of tight phase
synchronization across paths, an onerous task at mmW frequencies especially since
relative motion on the order of the wavelength (a few millimeters) can render the
optimal beamformer unuseable in practice. These conclusions on small losses with
directional beamforming as well as its robustness relative to the optimal scheme
provides a major fillip to the search for good directional learning approaches, a
task that has received significant and increasing attention in the literature.

\ignore{
With this choice, we have ${\sf SNR}_{\sf rx} =
\frac{\rho_{\sf prebf}}{N_t} \cdot {\bf f}^H {\sf H}^H {\sf H} {\bf f}$. To
maximize ${\bf f}^H {\sf H}^H {\sf H} {\bf f}$, we block decompose ${\bf f}$
and ${\sf H}^H {\sf H}$ as follows:
\begin{eqnarray}
{\sf H}^H {\sf H} = \left[
\begin{array}{cc}
{\sf A} & {\sf b} \\
{\sf b}^H & {\sf c}
\end{array}
\right] , \hspp \hspp {\bf f} = \left[
\begin{array}{c}
\underline{\alpha} \\
\alpha_{N_t} e^{j \phi_{N_t} }
\end{array}
\right],
\end{eqnarray}
where ${\sf A}$ is $(N_t - 1) \times (N_t - 1)$, ${\sf b}$ and $\underline{\alpha}$ are
$(N_t - 1) \times 1$, and ${\sf c}$ and $\alpha_{N_t}$ are scalars. With these choices,
${\bf f}^H {\sf H}^H {\sf H} {\bf f}$ can be simplified as
\begin{eqnarray}
{\bf f}^H {\sf H}^H {\sf H} {\bf f} =
\underline{\alpha}^H {\sf A} \underline{\alpha} +
2 \hsppp \alpha_{N_t} \cdot
{\sf Re} \left( e^{-j \phi_{N_t}} \cdot {\sf b}^H \underline{\alpha} \right)
+ \left( \alpha_{N_t} \right)^2 {\sf c}.
\label{eq_laststep_iter}
\end{eqnarray}
Clearly, the quantity in~(\ref{eq_laststep_iter}) is maximized by setting
$\phi_{N_t} = \angle{ {\sf b}^H \underline{\alpha} }$ and $\alpha_{N_t} =
\frac{1}{ \sqrt{N_t} }$.
}

\appendix

\subsection{Proof of Prop.~\ref{prop_evectors_HhermH}}
\label{app_prop_evectors_HhermH}
The $N_t \times N_t$ matrix ${\sf H}^H {\sf H}$ can be expanded as
\begin{eqnarray}
\frac{L}{N_t N_r} \cdot {\sf H}^H {\sf H} & = &
\sum_{i,j} 
\alpha_i^{\star} \alpha_j \cdot \left( {\bf u}_i^H {\bf u}_j \right) \cdot
{\bf v}_i {\bf v}_j^H
\label{eq_HhermH}
\\ & = &
{\sf V} \hsppp {\sf A} \hsppp {\sf V}^H
\nonumber
\end{eqnarray}
where ${\sf V} = \left[ \alpha_1^{\star} \hsppp {\bf v}_1, \hsppp \cdots , \hsppp
\alpha_L^{\star} \hsppp {\bf v}_L \right]$ and ${\sf A}(i,j) = {\bf u}_i^H {\bf u}_j,
\hsppp i, j = 1, \cdots, L$.
Let ${\sf X}$ be an $L \times L$ eigenvector matrix of ${\sf A} \hsppp
{\sf V}^H \hsppp {\sf V}$ with the corresponding diagonal matrix of eigenvalues
denoted by ${\sf D}$. That is (the eigenvalue equation is given as),
\begin{eqnarray}
\left( {\sf A} \hsppp {\sf V}^H \hsppp {\sf V} \right) \cdot {\sf X} =
{\sf X} \cdot {\sf D}.
\label{eq_premultiply_eigs}
\end{eqnarray}
Pre-multiplying both sides of~(\ref{eq_premultiply_eigs}) by ${\sf V}$, we have
\begin{eqnarray}
{\sf V} \hsppp {\sf X} \cdot {\sf D} =
\left( {\sf V} \hsppp {\sf A} \hsppp {\sf V}^H \hsppp {\sf V} \right) \cdot {\sf X} =
\left( \frac{L}{N_t N_r} \cdot {\sf H}^H {\sf H}  \right) \cdot
{\sf V} \hsppp {\sf X}.
\label{eq_premultiply_eigs1}
\end{eqnarray}
Reading equation~(\ref{eq_premultiply_eigs1}) from right to left, we see that
${\sf V} \hsppp {\sf X}$ forms the eigenvector matrix for ${\sf H}^H {\sf H}$ with the
diagonal eigenvalue matrix being the same as ${\sf D}$. In other words, all the
eigenvectors of ${\sf H}^H {\sf H}$ can be represented as {\em linear combinations}
of ${\bf v}_1, \cdots , {\bf v}_L$.
The only difference between the $L \leq N_t$ and $L > N_t$ cases is that the number
of distinct eigenvectors of ${\sf X}$ is less than or equal to $L$ and $N_t$ in the
two cases, respectively.

Given the structure of ${\bf f}_{\sf opt} = \sum_{j = 1}^L \beta_j {\bf v}_j$, we have
\begin{eqnarray}
{\sf H} \hsppp {\bf f}_{\sf opt} & = &
\left ( \sum_{i = 1}^L \alpha_i {\bf u}_i {\bf v}_i^H \right)
\cdot \left( \sum_{j = 1}^L \beta_j {\bf v}_j \right)
\\
& = & \sum_{i = 1}^L \alpha_i \cdot \left( \sum_j \beta_j {\bf v}_i^H {\bf v}_j \right)
{\bf u}_i
\label{eq43}
\end{eqnarray}
and thus ${\bf g}_{\sf opt}$ is a linear combination of $\{ {\bf u}_1, \cdots,
{\bf u}_L \}$.
\qed

\subsection{Proof of Prop.~\ref{prop_v1v2orth}}
\label{app_prop_v1v2orth}
A simple substitution of ${\bf v}_1^H {\bf v}_2 = 0$ leads to 
\begin{eqnarray}
\widetilde{ {\sf SNR} }_{\sf rx} =
\frac{1}{L} \cdot \left[
\beta^2 |\alpha_1|^2 + (1 - \beta^2) |\alpha_2|^2  + 2 \beta \sqrt{1 - \beta^2}
\cdot |\alpha_1 | | \alpha_2| \cdot | {\bf u}_1^{H} {\bf u}_2 | \cdot
\cos \left( \nu - \phi \right) \right],
\nonumber
\end{eqnarray}
which upon optimization over $\phi$ results in
\begin{eqnarray}
\widetilde{  {\sf SNR} }_{\sf rx} =
\frac{1}{L} \cdot \left[
\beta^ 2 |\alpha_1|^2 + (1 - \beta^2) |\alpha_2|^2
+ 2 \beta \sqrt{1 - \beta^2} \cdot
|\alpha_1 | | \alpha_2| \cdot | {\bf u}_1^{H} {\bf u}_2 | \right].
\nonumber
\end{eqnarray}
A straightforward computation shows that the optimal solution to the above optimization
in the $\beta$ variable satisfies the quadratic equation
\begin{eqnarray}
\beta^2 (1 - \beta^2) \cdot \left( |\alpha_1|^2 - |\alpha_2|^2 \right)^2 =
|\alpha_1|^2 |\alpha_2|^2 |{\bf u}_1^H {\bf u}_2|^2 \cdot \left( 2 \beta^2 - 1 \right)^2
\nonumber
\end{eqnarray}
and is of the form in~(\ref{eq_betaopt}). A straightforward substitution of
the structure of ${\bf f}_{\sf opt}$ in~(\ref{eq43}) results in
${\bf g}_{\sf opt}$.
\qed

\subsection{Proof of Prop.~\ref{prop_u1u2orth}}
\label{app_prop_u1u2orth}
When ${\bf u}_1$ and ${\bf u}_2$ are orthogonal, we have
\begin{eqnarray}
\widetilde{ {\sf SNR} }_{\sf rx}
& = & \frac{ A \cdot \left( |\alpha_1|^2 + |\alpha_2|^2 \right)
+ 2 \left( |\alpha_1|^2 + |\alpha_2|^2 \right) \cdot \beta \sqrt{1 - \beta^2}
\cdot | {\bf v}_1^{H} {\bf v}_2| \cdot \cos \left( \phi \right)
} {L \cdot \left( 1 + 2 \beta \sqrt{1 - \beta^2} \ | {\bf v}_1^{H} {\bf v}_2 |
\cos( \phi)  \right) }
\nonumber \\
& = & \left( \frac{ |\alpha_1|^2 + |\alpha_2|^2}{L} \right) \cdot
\frac{ A + 2 \beta \sqrt{1 - \beta^2} \cdot | {\bf v}_1^{H} {\bf v}_2|
\cdot \cos \left( \phi  \right) }
{ 1 + 2 \beta \sqrt{1 - \beta^2} \cdot | {\bf v}_1^{H} {\bf v}_2|
\cdot \cos \left( \phi \right) }
\nonumber
\end{eqnarray}
where
\begin{eqnarray}
A = \frac{ \beta^2 \left( |\alpha_1|^2 + |\alpha_2|^2 | {\bf v}_1^{H} {\bf v}_2|^2 \right)
+ (1 - \beta^2) \left( |\alpha_2|^2 + |\alpha_1|^2 | {\bf v}_1^{H} {\bf v}_2|^2 \right)}
{ |\alpha_1|^2 + |\alpha_2|^2 } \leq 1.
\nonumber
\end{eqnarray}
Since $A \leq 1$ for all choices of $\beta^2$, it is easy to see that
$\widetilde{ {\sf SNR} }_{\sf rx}$ is always maximized when $\cos \left( \phi
\right)$ is maximized at $1$ by the choice
$\theta = - \angle{ {\bf v}_1^{H} {\bf v}_2  }$. 
Using this fact, after some manipulations, we have the following:
\begin{eqnarray}
\widetilde{ {\sf SNR } }_{\sf rx} & = & 
\frac{ |\alpha_1|^2 \cdot \left( \beta + \sqrt{1 - \beta^2}
|{\bf v}_1^{H} {\bf v}_2| \right)^2 + |\alpha_2|^2 \cdot
\left( \beta |{\bf v}_1^{H} {\bf v}_2|
+ \sqrt{1 - \beta^2} \right)^2 }{ L \cdot \left( 1 + 2 \beta \sqrt{1 - \beta^2}
| {\bf v}_1^{H} {\bf v}_2 | \right) } \nonumber \\
& = & \frac{ |\alpha_1|^2 + |\alpha_2|^2 }{L }
- \left( \frac{ 1 - |{\bf v}_1^{H} {\bf v}_2|^2}{L} \right) \cdot
\left( \frac{ \beta^2 \cdot |\alpha_2|^2 + (1 - \beta^2) \cdot |\alpha_1|^2 }
{  1 + 2 \beta \sqrt{1 - \beta^2} \cdot |{\bf v}_1^{H} {\bf v}_2|  }
\right).
\nonumber
\end{eqnarray}
Thus, the optimal choice of $\beta^2$ (denoted as $\beta_{\sf opt}^2$) is that choice
that minimizes the quantity in the parentheses above.
It can be seen that this optimal choice satisfies the equation:
\begin{eqnarray}
\beta_{\sf opt}^2 \cdot |\alpha_2|^2 - (1 - \beta_{\sf opt}^2) \cdot
|\alpha_1|^2 = \beta_{\sf opt} \cdot \sqrt{ 1 - \beta_{\sf opt}^2 }
\cdot \frac{ \left( |\alpha_1|^2 - |\alpha_2|^2 \right) }
{ | {\bf v}_1^{H} {\bf v}_2| }
\nonumber
\end{eqnarray}
and is explicitly written as in the statement of the proposition.
A straightforward substitution of ${\bf f}_{\sf opt}$ in~(\ref{eq43}) results in
${\bf g}_{\sf opt}$.
\qed

\subsection{Proof of Prop.~\ref{prop_u1u2par}}
\label{app_prop_u1u2par}
In this setting, we can use the fact that $|{\bf u}_1^H {\bf u}_2| \approx 1$ to rewrite
${\widetilde{\sf SNR}}_{\sf rx}$ as
\begin{eqnarray}
{\widetilde{\sf SNR}}_{\sf rx} & = & \frac{ |\alpha_1|^2 + |\alpha_2|^2 +
\frac{ 2 |\alpha_1| |\alpha_2| \cos(\nu) }{ |{\bf v}_1^H {\bf v}_2|} }{L}
- \frac{ \left( 1 -  |{\bf v}_1^H {\bf v}_2|^2 \right) }{L} \cdot
f \left(\beta, \hsppp \phi \right)
\nonumber \\
f \left(\beta, \hsppp \phi \right)
& \triangleq & \frac{|\alpha_1|^2 \cdot ( 1 - \beta^2)
+ |\alpha_2|^2 \cdot \beta^2
+ \frac{ 2 |\alpha_1| |\alpha_2| \cos(\nu) }{ |{\bf v}_1^H {\bf v}_2|}
+ 2 \beta \sqrt{1 - \beta^2} \cdot |\alpha_1||\alpha_2| \cdot \cos(\nu + \phi) }
{ 1 + 2 \beta \sqrt{1 - \beta^2} \cdot | {\bf v}_1^H {\bf v}_2 | \cdot
\cos(\phi)  } .
\nonumber
\end{eqnarray}
To find the structure of ${\bf f}_{\sf RSV}$, we need to find $\{ \beta^{\star}, \hsppp \phi^{\star} \}
= \arg \min \limits_{ \beta, \hsppp \phi } f \left( \beta , \hsppp \phi \right)$.

We now claim that for any $\nu, \beta, K = \frac{ |\alpha_1| }{ |\alpha_2| } \geq 1$,
$\phi^{\star} = \nu$. Substituting $\phi^{\star} = \nu$, we have
\begin{eqnarray}
f \left(\beta, \hsppp \phi^{\star} \right) & = &
\frac{|\alpha_1|^2 \cdot ( 1 - \beta^2)
+ |\alpha_2|^2 \cdot \beta^2
+ \frac{ 2 |\alpha_1| |\alpha_2| \cos(\nu) }{ |{\bf v}_1^H {\bf v}_2|}
+ 2 \beta \sqrt{1 - \beta^2} \cdot |\alpha_1||\alpha_2| \cdot \cos(2\nu ) }
{ 1 + 2 \beta \sqrt{1 - \beta^2} \cdot | {\bf v}_1^H {\bf v}_2 | \cdot
\cos(\nu)  }
\nonumber \\
& = & \frac{ 2  |\alpha_1| |\alpha_2| \cos(\nu) }{ | {\bf v}_1^H {\bf v}_2|}
+ \frac{ \left( |\alpha_1| \sqrt{1 - \beta^2} - |\alpha_2| \beta \right)^2 }
{ 1 + 2 \beta \sqrt{1 - \beta^2} \cdot | {\bf v}_1^H {\bf v}_2 | \cdot
\cos(\nu)  }
\nonumber \\
\Longrightarrow
{\widetilde{\sf SNR}}_{\sf rx} & = & \frac{ |\alpha_1|^2 + |\alpha_2|^2 +
2 |\alpha_1| |\alpha_2| \cos(\nu)\cdot |{\bf v}_1^H {\bf v}_2|  }{L}
\nonumber \\
& & {\hspace{0.5in}}
- \frac{ \left( 1 -  |{\bf v}_1^H {\bf v}_2|^2 \right) }{L} \cdot
\min \limits_{\beta}
\frac{ \left( |\alpha_1| \sqrt{1 - \beta^2} - |\alpha_2| \beta \right)^2 }
{ 1 + 2 \beta \sqrt{1 - \beta^2} \cdot | {\bf v}_1^H {\bf v}_2 |
\cos(\nu)  }.
\nonumber
\end{eqnarray}
By setting $\beta^2 = \beta_{\sf opt}^2$, the last term can be minimized (to $0$),
resulting in:
\begin{eqnarray}
{\widetilde{\sf SNR}}_{\sf rx} & = & \frac{ |\alpha_1|^2 + |\alpha_2|^2 +
2 |\alpha_1| |\alpha_2| \cos(\nu)\cdot |{\bf v}_1^H {\bf v}_2|  }{L}.
\nonumber
\end{eqnarray}
\qed

\bibliographystyle{IEEEbib}
\bibliography{newrefsx}

\end{document}